\documentclass[journal,twoside,web]{ieeecolor}
\usepackage{generic}
\usepackage{cite}
\usepackage{textcomp}
\usepackage{subcaption} 
\usepackage{booktabs} 
\usepackage{multirow}
\usepackage[T1]{fontenc}
\usepackage[utf8]{inputenc}
\usepackage{pgfplots}
\usepackage{grffile}
\pgfplotsset{compat=newest}
\usetikzlibrary{plotmarks}
\usetikzlibrary{arrows.meta}
\usepgfplotslibrary{patchplots}
\definecolor{matlabBlue}{rgb}{0,0.4470,0.7410}
\definecolor{matlabBlack}{rgb}{0,0,0}

\usepackage{multirow}
\usepackage{booktabs}
\usepackage{tablefootnote}   

\usepackage{amssymb,latexsym,amsfonts,amsmath,bbm}
\usepackage{mathrsfs}
\usepackage{graphicx}
\usepackage{paralist}
\usepackage{dsfont}
\usepackage{eso-pic}
\usepackage{epsfig}
\usepackage{mathtools}
\usepackage{epstopdf}
\usepackage{lipsum}
\usepackage{cite}

\usepackage[nocomma]{optidef}

\usepackage{tikz}
\usetikzlibrary{calc,shapes,arrows}

\usepackage{algorithm}
\usepackage{algorithmic}

\newcounter{theorem}
\newcounter{definition}
\newcounter{lemma}
\newcounter{claim}
\newcounter{problem}
\newcounter{proposition}
\newcounter{corollary}
\newcounter{construction}
\newcounter{example}
\newcounter{xca}
\newcounter{comments}
\newcounter{remark}
\newcounter{assumption}

\newtheorem{theorem}[theorem]{Theorem}
\newtheorem{lemma}[lemma]{Lemma}

\newtheorem{definition}[definition]{Definition}

\newtheorem{remark}[remark]{Remark}

\numberwithin{equation}{section}

\DeclareFontFamily{U}{stix2bb}{}
\DeclareFontShape{U}{stix2bb}{m}{n} {<-> stix2-mathbb}{}

\usepackage[many]{tcolorbox}
\usetikzlibrary{calc}
\tcbuselibrary{skins}

\newtcolorbox{resp}[1][]{%
	enhanced jigsaw,%
	colback=gray!5!white,%
	colframe=gray!80!black,%
	size=small,%
	boxrule=1pt,%
	halign title=flush center,%
	coltitle=black,%
	breakable,%
	drop shadow=black!50!white,%
	attach boxed title to top left={xshift=1cm,yshift=-\tcboxedtitleheight/2,yshifttext=-\tcboxedtitleheight/2},%
	minipage boxed title=3cm,%
	boxed title style={%
		colback=white,%
		size=fbox,%
		boxrule=1pt,%
		boxsep=2pt,%
		underlay={%
			\coordinate (dotA) at ($(interior.west) + (-0.5pt,0)$);
			\coordinate (dotB) at ($(interior.east) + (0.5pt,0)$);
			\begin{scope}[gray!80!black]
				\fill (dotA) circle (2pt);
				\fill (dotB) circle (2pt);
			\end{scope}
		}%
	},%
	#1%
}

\newcommand{\CLbr}{\mathcal{{C}}_1\, \mathcal{R}_1(x,x_h)}
\newcommand{\DLbr}{\mathcal{{C}}_2\, \mathcal{R}_2(x,x_h)}
\DeclareRobustCommand\sampleline[1]{%
	\tikz\draw[#1] (0,0) (0,\the\dimexpr\fontdimen22\textfont2\relax)
	-- (2em,\the\dimexpr\fontdimen22\textfont2\relax);%
}

\usepackage{xspace}
\usepackage{subcaption} 

\def\BibTeX{{\rm B\kern-.05em{\sc i\kern-.025em b}\kern-.08em
		T\kern-.1667em\lower.7ex\hbox{E}\kern-.125emX}}

\usepackage{etoolbox}

\makeatletter
\AtBeginDocument{%
	\pagestyle{headings}%
	
	\patchcmd{\@oddhead}{\\[-19pt]}{\\[-8pt]}{}{}%
	\patchcmd{\@evenhead}{\\[-19pt]}{\\[-8pt]}{}{}%
}
\makeatother

\definecolor{blue(ryb)}{rgb}{0.01, 0.28, 1.0}
\definecolor{fashionfuchsia}{rgb}{0.96, 0.0, 0.63}
\makeatletter
\let\NAT@parse\undefined
\makeatother

\definecolor{cerise}{rgb}{0.87, 0.19, 0.39}

\usepackage[colorlinks=true,
citecolor=cerise,
linkcolor=cerise,
urlcolor=cerise,    
final             
]{hyperref}

\renewcommand{\theequation}{\arabic{equation}}
\counterwithout*{equation}{section}

\makeatletter
\def\@opargbegintheorem#1#2#3{\textit{#1\ #2} \textit{(#3):}}
\makeatother

\usepackage{microtype}

\begin{document}
	
	\title{A Data-Driven Krasovskii-Based Approach for Safety Controller Design of Time-Delayed Uncertain Polynomial Systems}
	\author{ \IEEEmembership{}	
		\thanks{}
	}
	
	\author{Omid Akbarzadeh, \IEEEmembership{Student Member,~IEEE}, MohammadHossein Ashoori, \IEEEmembership{Student Member,~IEEE}, Amy Nejati,  \IEEEmembership{Senior Member,~IEEE}, and 	Abolfazl Lavaei, \IEEEmembership{Senior Member,~IEEE}
		\thanks{All the authors are with the School of Computing, Newcastle University, United Kingdom. Email:
			\texttt{o.akbarzadeh2@newcastle.ac.uk},
			\texttt{m.ashoori2@newcastle.ac.uk},
			\texttt{amy.nejati@newcastle.ac.uk},
			\texttt{abolfazl.lavaei@newcastle.ac.uk}. }
	}
	
	\maketitle
	\begin{abstract}
		We develop a data-driven framework for the synthesis of \emph{robust Krasovskii control barrier certificates} (RK-CBC) and corresponding robust safety controllers (R-SC) for discrete-time input-affine uncertain polynomial systems with {(partially) unknown} dynamics, while explicitly accounting for \emph{unknown-but-bounded disturbances} and \emph{time-invariant delays} using only observed input–state data. Although control barrier certificates have been extensively studied for safety analysis of control systems, existing work on unknown systems with time delays, particularly in the presence of disturbances, remains limited. The challenge of safety synthesis for such systems stems from two main factors: first, the system’s mathematical model is unavailable; and second, the safety conditions should explicitly incorporate the effects of time delays on system evolution during the synthesis process, while remaining robust to unknown disturbances. To address these challenges, we develop a data-driven framework based on {RK-CBC}, extending the classical CBC formulation for delay-free systems to explicitly account for time delays by aggregating delayed components within the barrier construction. The proposed framework relies solely on input–state data collected over a finite time horizon, enabling the direct synthesis of RK-CBC and R-SC from observed trajectories without requiring an explicit system model. The synthesis is cast as a \emph{data-driven} sum-of-squares (SOS) {feasibility} program, yielding a structured design methodology. As a result, robust safety is guaranteed in the presence of unknown disturbances and time delays {as long as the closed-loop trajectories remain within the state set}. The effectiveness of the proposed method is demonstrated through three case studies, including two \emph{physical systems}.
	\end{abstract}
	
	
	\section{Introduction}\label{sec: intro}
	\IEEEPARstart{T}{IME} delays are an intrinsic feature of modern control {systems, arising from latencies in sensing, actuation, computation, scheduling, and communication over shared networks, as well as from transport and memory effects within the plant}. Such delays are pervasive in application domains including communications~\cite{Srikant2004} and process control~\cite{Seborg2004}, where systems operate under strict safety requirements. As control architectures increasingly rely on networked and shared communication infrastructures, the prevalence of safety-critical systems subject to time delays is expected to grow, further amplifying the need for rigorous safety analysis and controller synthesis methods that explicitly account for delayed system behavior. The presence of delays, while frequently unavoidable, can fundamentally alter system behavior, leading to reduced stability margins, degraded performance, and potential violations of safety constraints that would otherwise be satisfied in delay-free settings.
	
	{Delays are commonly classified by where they enter the system~\cite{Richard2003,Fridman2014}. \emph{Input} and \emph{output} delays affect the applied control input or the measurements and are typically handled by prediction-based compensation~{\cite{Artstein1982,Krstic2009}}, whereas \emph{state} delays, arising from delayed couplings, transport, or memory effects intrinsic to the plant~{\cite{Gu2003,Kyrychko2010}}, alter the dynamics themselves. This work considers time-invariant state delays in discrete-time polynomial dynamics; the delayed state moreover enters the closed loop through the synthesized delayed-state feedback.}
	
	In spite of these effects, safety analysis and controller synthesis for uncertain discrete-time nonlinear systems with unknown models and time delays remain relatively underexplored and technically demanding. The difficulty of this problem is inherently multifaceted: beyond the absence of an explicit mathematical description, the system is subject to uncertainties and delayed dynamics that directly affect its evolution. These factors complicate both the formulation of safety conditions and the synthesis of controllers aimed at providing formal guarantees.
	
	\begin{figure}[tbp]
		\centering
		\includegraphics[width=0.95\linewidth]{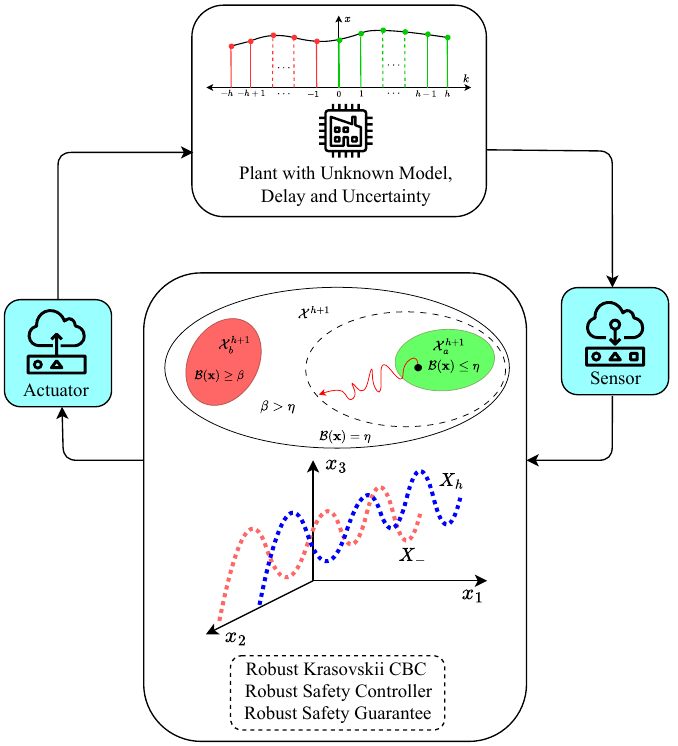}
		\caption{{Architecture of the proposed data-driven safety controller synthesis
					framework. A single input-state trajectory is collected from the unknown
					delayed polynomial system, yielding the data sets $X_{-}$ and $X_{h}$, which
					contain the measured current and $h$-step delayed states, respectively. These data, together with the admissible
					disturbance bound, are used to formulate an SOS feasibility program that
					simultaneously synthesizes an RK-CBC
					and a polynomial safety controller. The resulting R-SC, which feeds
					back both the current and delayed states, guarantees that all closed-loop
					trajectories initialized in the initial-history set $\mathcal{X}_a^{h+1}$ remain outside the unsafe region $\mathcal{X}_b$ while evolving within the state set $\mathcal{X}$, without requiring knowledge of the system matrices.}
			\label{fig:Sys}}
	\end{figure}
	To address safety controller design in the context of uncertain control systems without time delays, a significant body of literature employs control barrier certificates (CBC), initially introduced in~\cite{prajna2004safety}. Similar to Lyapunov functions, CBC impose conditions on both the certificate and its evolution along system trajectories. By selecting an initial level set of CBC that corresponds to a specified set of initial states, one can create a separation between the unsafe region and all admissible evolutions, thereby providing (probabilistic) safety guarantees~\cite{santoyo2021barrier,nejati2024context,borrmann2015control,zaker2024compositional,ames2019control,lavaei2022automated}. While promising, the approaches in this body of work predominantly assume delay-free system dynamics, and their direct extension to systems with time delays is nontrivial.
	
	In the model-based setting, several studies have explored safety guarantees for systems with delayed dynamics. The work in~\cite{Jadbabaie_delay} generalizes barrier certificates to time-delay systems through the use of Razumikhin- and Krasovskii-type functionals for deterministic nonlinear polynomial models, with an emphasis on safety verification rather than controller synthesis. Complementarily,~\cite{Ames_Safety} proposes safety functionals that certify system safety without relying on delay approximations for deterministic nonlinear systems. Building on this line of research,~\cite{Ames_Safety_1} introduces control barrier functionals for nonlinear systems with state delays, enabling the construction of safety filters. More recently,~\cite{liu2023safety} extends barrier-functional approaches to time-delay systems affected by disturbances by incorporating input-to-state safety concepts. In contrast to these prior works, this paper addresses safety synthesis for \emph{unknown} polynomial systems subject to both time delays and unknown-but-bounded disturbances, a setting that has not been explicitly considered in the existing literature. It is also worth noting that a substantial body of literature has examined the stability of control systems subject to time delays; however, a detailed discussion of these results is beyond the scope of this paper (see \emph{e.g.,}~\cite{Papachristodoulou_Convex,REN2022110563,Fridman}).
	
	Since the design of CBC in the aforementioned literature relies on accurate knowledge of {the} system’s mathematical dynamics, often rendering model-based approaches impractical, data-driven techniques have gained increasing attention by enabling analysis directly from system measurements without requiring explicit models. These methods are mainly divided into \emph{indirect} and \emph{direct} approaches, each featuring distinct methodologies and application areas~\cite{Hou2013model,dorfler2022bridging,nejati2023data}. In particular, \emph{indirect} data-driven methods involve system identification, leveraging a wide array of robust tools from model-based control techniques following the identification phase. However, these methods face challenges related to computational complexity in two critical phases: first, during the model identification process, and later in solving the model-based problem.  Conversely, \emph{direct} data-driven approaches bypass the system identification step, applying system measurements directly to analyze unknown systems~\cite{dorfler2022bridging}. 
	
	Over the past few years, data-driven approaches have emerged as a powerful alternative for the safety certification of systems with unknown dynamics (\emph{e.g.,}~\cite{bisoffi2020controller,nejati2022data}). A recent study~\cite{Akbarzadeh_LCSS} advances direct data-driven methods for safety synthesis of discrete-time polynomial systems; however, its scope is limited to delay-free systems without uncertainties. In contrast, the framework proposed in this work establishes \emph{robust} safety guarantees for systems subject to unknown-but-bounded disturbances {in the presence of} time-invariant delays{, which requires a different certificate structure, a two-channel data parameterization, and robust convexification rather than a straightforward extension of~\cite{Akbarzadeh_LCSS}}. Related data-driven studies~\cite{Fridman-data,Kong} focus on the synthesis of state-feedback controllers ensuring \emph{stability} for \emph{linear} discrete-time systems with delays, using noisy and noise-free data, respectively. By comparison, our work addresses \emph{safety} synthesis for discrete-time \emph{nonlinear polynomial} systems with time-invariant delays and \emph{unknown-but-bounded disturbances}.
	
	Scenario-based methodologies~\cite{calafiore2006scenario} address safety problems for unknown systems by leveraging data, often through the use of chance constraints~\cite{esfahani2014performance,margellos2014road}. While these approaches can provide formal probabilistic guarantees, they rely on the assumption that data samples are independent and identically distributed (i.i.d.), which implies that only an input–output pair can be extracted from each trajectory~\cite{calafiore2006scenario}. As a consequence, achieving a desired confidence level typically requires collecting a large number of independent trajectories (see~\cite{nejati2023data} for safety analysis without delays). In contrast, the proposed data-driven framework does not require i.i.d.\ samples and can perform safety analysis solely using a single trajectory collected over a finite time horizon. Moreover, the resulting certificates are deterministic (\emph{i.e.,} confidence $1$) rather than probabilistic, unlike scenario-based methods that provide guarantees with confidence levels between zero and one. These advantages come with inherent trade-offs: scenario-based approaches apply to general Lipschitz-continuous nonlinear dynamics, whereas our framework is tailored to nonlinear systems with polynomial dynamics.
	
	\textbf{Original contributions.} We develop a data-driven framework for safety controller design of discrete-time input-affine polynomial systems with {(partially) unknown} dynamics, accounting for both unknown-but-bounded disturbances and time-invariant delays. Our approach enables the direct synthesis of {robust Krasovskii control barrier certificates (RK-CBC)} and corresponding {robust safety controllers (R-SC)} from {a single} finite-horizon input–state {trajectory}. {To the best of our knowledge, no existing result provides safety controller synthesis with deterministic guarantees for this class: model-based delay-safety results require the dynamics~\cite{Jadbabaie_delay,Ames_Safety,Ames_Safety_1,liu2023safety}, data-driven safety results address delay-free systems~\cite{Akbarzadeh_LCSS,akbarzadeh2024learning}, and data-driven results for delayed systems address the stability of linear systems~\cite{Fridman-data,Kong}.} We extend the classical CBC framework, originally developed for delay-free systems, to time-delayed systems by aggregating delay terms within {an RK-CBC}, with a shared matrix and geometric delay weights, co-designed with a two-channel data-driven parameterization of the delayed dynamics. This {co-design enables safety guarantees for delayed systems that are robust to disturbances and valid whenever the closed-loop trajectory remains within the prescribed state set.} It also reduces the set-product complexity from \((h+1)\)-fold products (with $h$ being the delay), as encountered in model-based conditions for delayed systems, to at most a \emph{two-fold} product, {so that the size of the resulting sum-of-squares (SOS) program is independent of the delay}. {The effectiveness of the approach is demonstrated through three case studies, two of them physical systems, and a comparison with stabilizing designs that may violate safety, while the proposed R-SC maintains it.} A high-level overview of the proposed framework is illustrated in Fig.~\ref{fig:Sys}.
	
	A limited subset of this work was recently presented in~\cite{akbarzadeh2024learning}, which addresses delay-free systems. The data-driven framework developed in the present paper extends~\cite{akbarzadeh2024learning} in two fundamental directions. First, we consider a broader class of systems, namely uncertain discrete-time input-affine polynomial dynamics with \emph{time-invariant {state} delays} and delayed-state feedback. Second, to explicitly capture these effects, we introduce {an RK-CBC} that aggregates delay terms and incorporates them directly into the certificate construction. These extensions {require a two-channel data parameterization and joint convexification, and} introduce significant technical challenges in both barrier certificate derivation and safety controller synthesis. Unlike~\cite{akbarzadeh2024learning}, the proposed framework establishes safety guarantees that are robust to disturbances {in the presence of delays}.
	
	\section{System Description}\label{sec:Description}
	\subsection{Notation and Preliminaries}
	The sets of real numbers, non-negative and positive real numbers are denoted by $\mathbb{R}$, $\mathbb{R}^+_0$, and $\mathbb{R}^+$\!, respectively. Similarly, $\mathbb{N} = \{0, 1, 2, \ldots\}$ and $\mathbb{N}^+ = \{1, 2, \ldots\}$ represent the sets of non-negative and positive integers, respectively. The identity matrix of size $n \times n$ is expressed as $\mathds{I}_n$. Additionally, an $n$-dimensional vector and an {$n\times m$ matrix with all zero entries are represented by $\boldsymbol{0}_n$ and $\boldsymbol{0}_{n\times m}$, respectively.} Given \( N \) scalars \( x_i \in \mathbb{R} \), the \emph{column} vector \( x \) is formed by stacking  as \( x = [x_1; \dots; x_N] \). Moreover, we use $[x_1\;\; \ldots\;\; x_N]$ to represent the horizontal concatenation of vectors $x_i\in\mathbb{R}^n$, forming an $n\times N$ matrix, and $[A_1\;\; \ldots\;\; A_N]$ to denote the horizontal concatenation of matrices $A_i \in \mathbb{R}^{n \times m_i}$, yielding an $n \times \sum_{i=1}^N m_i$ matrix.
	For matrices $A_i \in \mathbb{R}^{m_i \times n_i}$, we write $\mathsf{blkdiag}(A_1,\dots,A_N)$ for the block-diagonal matrix with diagonal blocks $A_1,\dots,A_N$ and all off-diagonal blocks equal to zero (of compatible sizes), \emph{i.e.,} $\mathsf{blkdiag}(A_1,\dots,A_N)\in\mathbb{R}^{\left(\sum_{i=1}^N m_i\right)\times\left(\sum_{i=1}^N n_i\right)}.$ A \emph{symmetric} matrix $P$ is positive definite when $P \succ 0$, while $P \succeq 0$ signifies that $P$ is a \emph{symmetric} positive semi-definite matrix. The transpose of a matrix $P$ is written as $P^\top$. {For a matrix $A$, $\operatorname{col}(A)$ denotes its column space, \emph{i.e.,} the span of its columns.} The symbol $\otimes$ denotes the Kronecker product; for $A \in \mathbb{R}^{m \times n}$ and $B \in \mathbb{R}^{p \times q}$, the matrix $A \otimes B \in \mathbb{R}^{mp \times nq}$ is formed by replacing each entry $a_{ij}$ of $A$ with the block $a_{ij}B$.
	We use $\Vert \cdot \Vert$ to denote the Euclidean norm for a vector $x \in \mathbb{R}^n$ and the induced 2-norm for a matrix $A \in \mathbb{R}^{n \times m}$. The empty set is denoted by $\emptyset$. In a symmetric matrix, the symbol $(\star)$ denotes the entry given by the transpose of the corresponding symmetric position.
	For a system $\Upsilon$ and a property $\Theta$, the notation $\Upsilon \models \Theta$ signifies that $\Upsilon$ satisfies $\Theta$.
	The Cartesian product of two sets $\mathcal{A}$ and $\mathcal{B}$ is denoted by $\mathcal{A} \times \mathcal{B}$. For a positive integer $n \in \mathbb{N}^+$, the notation $\mathcal{A}^{n+1}$ represents the $(n+1)$-fold Cartesian product of $\mathcal{A}$ with itself. For the set $\mathcal{V} \subseteq \mathcal{A}$, the expression $\mathcal{A} \backslash \mathcal{V}$ denotes the set of elements in $\mathcal{A}$ not in $\mathcal{V}$.
	
	\subsection{Discrete-Time Input-Affine Time-Delayed Uncertain Polynomial System}
	We begin by introducing discrete-time input-affine polynomial systems with time-invariant delays and bounded disturbances, which serve as the primary model for the derivation of robust safety certificates and corresponding robust safety controllers.
	
	\begin{definition}[\textbf{dt-IAUPS-td}]\label{def: dt-NPS}
		A discrete-time input-affine uncertain polynomial system with time-invariant
		delay (dt-IAUPS-td) is defined by
		{\begin{align}\notag
				\Upsilon\!: x(k\!+\!1)  &\!= A_1  \mathcal M (x(k)\!) + A_2  \mathcal M (x(k-h))\\\label{eq: dt-NPS} &~~~  +\! B \mathcal G (x(k),x(k-h))u(k) \!+\! w(k),
		\end{align}}
		where 
		\begin{itemize}
			\item  $h \in \mathbb N^{+}$ is the time-invariant delay;
			\item $x(k), x(k-h) \in \mathcal{X}$, with $k \in \mathbb N$, denote the current and delayed states of the system, respectively, while $\mathcal{X} \subseteq \mathbb{R}^n$ is the state space;
			\item $\mathbf{x}_0 = (x(0), \ldots,x(-h)) \in \mathcal{X}^{h+1}$ is the given initial state history;
			{\item $A_1, A_2 \in \mathbb{R}^{n \times M}$ are system matrices, while $B \in \mathbb{R}^{n \times N}$ is the input matrix;}
			\item $\mathcal{M}: \mathcal{X} \to \mathbb{R}^M$, with $\mathcal{M}(\boldsymbol{0}_n) = \boldsymbol{0}_M$, maps the state to a vector of monomials, while $\mathcal{G}: \mathcal{X}^2 \to \mathbb{R}^{N \times m}$ maps the current and delayed states to a matrix of monomials;
			{\item $u(k) \in \mathcal{U} = \mathbb{R}^m$ is the control input;}
			\item $w(k)$ denotes the disturbance, which is assumed, for all $k$, to lie in
			\begin{equation}\label{eq:Wdelta}
				\mathcal W(\delta) = \big\{\, w \in \mathbb R^n \;\big|\; \|w\|^2 \le \delta \,\big\}, 
				\qquad \delta \in \mathbb R^{+}_{0}.
			\end{equation}		
			
		\end{itemize}
		{We represent the dt-IAUPS-td in~\eqref{eq: dt-NPS} by the tuple $\Upsilon = \big(A_1, A_2, B, \mathcal G, \mathcal M, \mathcal X, \mathcal U, \mathcal W, h\big)$. We denote by $\mathbf{x}_{\mathbf{x}_0uw}(k)$}
		the state trajectory of the dt-IAUPS-td $\Upsilon$ at time $k \in \mathbb{N}$, representing the state together with its $h$-step history, under input and disturbance {sequences $u(\cdot)$ and $w(\cdot)$,} starting from an initial state history $\mathbf{x}_0 \in \mathcal{X}^{h+1}$.
	\end{definition}
	We consider a state-feedback controller structured as
	\begin{equation}\label{feedback}
		u(k)= \mathcal{F}_1(x(k), x(k-h)) x(k) + \mathcal{F}_2(x(k), x(k-h)) x(k-h),
	\end{equation}
	where $ \mathcal F_1(x(k), x(k-h)) \in \mathbb R^{ m \times n}$ and $  \mathcal F_2(x(k), x(k-h)) \in \mathbb R^{m \times n}$ are \emph{polynomial} matrices to be designed within the proposed framework. In this work, the system and input matrices $A_1$, $A_2$, and $B$ are all assumed to be unknown. In addition, the disturbance $w$ is unknown but bounded, with a known bound $\delta$ as defined in~\eqref{eq:Wdelta}. While the exact forms of $\mathcal{M}(x)$ and $\mathcal{G}(x{,x_h})$, with $x:=x(k)$ and $x_h:=x(k-h)$, are not available, we assume that a suitable dictionary (\emph{i.e.,} a family of basis functions) can be constructed that is sufficiently expressive to capture the true system dynamics. This is achieved by first determining upper bounds on the maximum degrees of $\mathcal{M}(x)$ and $\mathcal{G}(x,x_h)$ based on physical insights into the system. These bounds allow the construction of $\mathcal{M}(x)$ and $\mathcal{G}(x,x_h)$ so as to include all possible combinations of state variables up to those specified degrees (cf.~the case study in Subsection~\ref{Case_study_1}). With a slight abuse of notation, we use $\mathcal{M}(x)$ and $\mathcal{G}(x,x_h)$ interchangeably to denote both the original and the extended dictionaries throughout the paper. {In this sense, the system is \emph{(partially) unknown}: $A_1$, $A_2$, $B$ are unknown, while a polynomial dictionary with known degree bounds is presumed, as is standard in data-driven polynomial control~\cite{guo2022datadriven}.}
	
	{\begin{remark}\label{rem:dist}
			The chosen dictionary only needs to contain the true dynamics: including all monomials up to a known degree bound preserves all guarantees, since redundant monomials simply carry zero coefficients in $A_1$, $A_2$, and $B$, at the cost of larger synthesis conditions and requiring sufficiently informative data. Moreover, if the true dynamics deviate from the dictionary span by an additive residual $\mathcal{J}(x,x_h)$ with $\|\mathcal{J}(x,x_h)\|^2 \le \vartheta$ on $\mathcal{X}^2$ (which holds for continuous $\mathcal{J}$ when $\mathcal{X}$ is compact) and the external disturbance satisfies $\|w\|^2 \le \delta_w$, then~\eqref{eq:Wdelta} holds for their sum with $\delta = (\sqrt{\delta_w}+\sqrt{\vartheta})^2$, preserving all guarantees at the price of conservatism. Conversely, if such a residual is not covered by $\delta$, the synthesis problem is typically (though not necessarily) infeasible, and any controller it might still return carries no formal guarantee for the true system.
	\end{remark}}
	
	{\begin{remark}\label{rem:delta}
			In~\eqref{eq: dt-NPS}, $w$ represents an external process disturbance (\emph{e.g.,} unmodeled forces, load variations, quantization, or network-induced perturbations~\cite{Donkers2012}) subject only to the deterministic bound~\eqref{eq:Wdelta}. As is standard in set-membership and data-driven robust control~\cite{bisoffi2020controller}, only an upper bound $\delta$ is required, which can be obtained from physical or hardware specifications for the disturbance source independently of the system dynamics. Over-estimating $\delta$ preserves all guarantees at the cost of increased conservatism, whereas under-estimating it invalidates the guarantees and is only partially detectable through synthesis infeasibility; hence, $\delta$ should be chosen with a safety margin above the best available estimate.
	\end{remark}}
	
	In the following, we formally define the safety specification for the dt-IAUPS-td in~\eqref{eq: dt-NPS}.
	
	\begin{definition}[\textbf{Robust Safety Property}]\label{safety}
		Given a {dt-IAUPS-td  $\Upsilon = \big(A_1, A_2, B, \mathcal G, \mathcal M, \mathcal X, \mathcal U, \mathcal W, h\big)$ as in} Definition~\ref{def: dt-NPS} subject to known time-invariant delays, consider a safety specification $\Theta=\left(\mathcal X_a, \mathcal X_b \right)$, where $\mathcal X_a, \mathcal X_b \subset \mathcal X$ are initial and unsafe sets, respectively, with $\mathcal X_a \cap \mathcal X_b=\emptyset$. The dt-IAUPS-td $\Upsilon$  is said to be  robustly safe against the unknown-but-bounded disturbance, denoted by $\Upsilon \models \Theta$, if any trajectory of $\Upsilon$ initialized from $\mathbf{x}_0 \in \mathcal{X}_a^{h+1}$ avoids entering $\mathcal{X}_b$ {for all $k\in\mathbb{N}$ prior to the first exit of the trajectory from $\mathcal{X}$}. Equivalently, for any {such} $k$, the state history $(x(k),\dots,x(k-h))$ should never lie in $\mathcal{X}_b \times (\mathcal{X}\backslash\mathcal{X}_b)^h$.
	\end{definition}
	\begin{remark}
		Simply excluding the state history from the set $\mathcal X_b^{h+1}$ is not sufficient to ensure safety, as this condition only eliminates trajectories for which all $h+1$ states simultaneously belong to the unsafe region, while still allowing situations in which a subset of the states may be unsafe. Instead, enforcing that the state history never enters the set $\mathcal X_b \times (\mathcal X \setminus \mathcal X_b)^h$ provides a safety guarantee: whenever the preceding $h$ states are safe, the current state is also prevented from entering the unsafe set. Since the initial state history is assumed to lie in $\mathcal X_a^{h+1}$ and $\mathcal X_a \cap \mathcal X_b = \emptyset$, enforcing this condition ensures that the system trajectory remains outside the unsafe region, thereby ensuring safety {(cf.\ Section~\ref{sec: KCBC})}.
	\end{remark}
	We define the state history at time $k \in \mathbb{N}$ as
	\begin{equation}\label{history}
		\mathbf{x}(k) =(x(k),x(k-1),\ldots,x(k-h)) \in \mathcal{X}^{h+1},   
	\end{equation}
	and the successor sequence is given by
	\begin{equation*}
		{\mathbf{x}^+:=\,}\mathbf{x}(k+1)=(x(k+1),x(k),\ldots,x(k-h+1))\in \mathcal{X}^{h+1}.
	\end{equation*}
	To facilitate the subsequent analysis and improve readability, we define, based on~\eqref{history},
	\begin{equation}\label{history-simple}
		\underbrace{\mathbf{x}(k)}_\mathbf{x}=(\underbrace{x(k)}_{x},\underbrace{x(k-1)}_{x_1},\ldots,\underbrace{x(k-h)}_{x_h}),   
	\end{equation}
	where $x$ represents the current state, $x_1$ the state delayed by one time step, and $x_h$ the state delayed by $h$ time steps.
	
	\section{Robust Krasovskii Control Barrier Certificates }\label{sec: KCBC}
	To ensure the safety of the dt-IAUPS-td in the presence of time-invariant delays and bounded disturbances, we introduce the notion of {RK-CBC}, defined as follows.
	
	\begin{definition}[\textbf{RK-CBC}]\label{def: CBC}
		Consider a {dt-IAUPS-td
			$\Upsilon = \big(A_1, A_2, B, \mathcal G, \mathcal M, \mathcal X, \mathcal U, \mathcal W, h\big)\!,$} with $\mathcal{X}_a \subset \mathcal X$  and $\mathcal{X}_b \subset \mathcal X$ being its initial and unsafe sets, respectively. Assuming the existence of constants $\eta, \beta, \gamma \in \mathbb{R}^{+}$, with $\beta > \eta$, and $\lambda \in (0,1)$, a function $\mathcal B: {\mathbb{R}^{n(h+1)}} \to \mathbb{R}_0^+$ is called a robust Krasovskii control barrier certificate (RK-CBC) for $\Upsilon$ if
		\begin{subequations}\label{eq: CBC}
			\begin{align}
				&  \:\:  \mathcal B(\mathbf{x}) \leq \eta, \hspace{3cm}  \forall \mathbf{x} \in \mathcal{X}^{h+1}_{a},\label{subeq: initial}\\
				&  \:\:  \mathcal B(\mathbf{x}) \geq  \beta, \hspace{3cm} \forall \mathbf{x} \in \mathcal{X}_b \times (\mathcal{X}\backslash\mathcal{X}_b)^h, \label{subeq: unsafe}
			\end{align}  
			and {$\forall \mathbf{x} \in  \mathcal{X}^{h+1}$, {$\exists u \in \mathcal{U}$,}}
			such that $\forall w \in \mathcal{W}(\delta)$
			\begin{align}\label{subeq: decreasing}
				&  {\mathcal B(\mathbf{x}^+) -  \lambda\mathcal B(\mathbf{x})} \leq  \gamma \|w\|^2,
			\end{align}
			with $\gamma$ satisfying
			\begin{align}\label{eq: climit}
				\gamma\delta \leq   (1-\lambda)\beta.
			\end{align}
			{Accordingly, $u$ {of the form~\eqref{feedback}} enforcing \eqref{subeq: decreasing} is a robust safety controller (R-SC) for the dt-IAUPS-td.}
		\end{subequations}
	\end{definition}
	
	{Intuitively, condition~\eqref{subeq: initial} requires $\mathcal{B}$ to be at most $\eta$ on every initial history, while~\eqref{subeq: unsafe} requires $\mathcal{B}$ to be at least $\beta>\eta$ on every history corresponding to a first entry into $\mathcal{X}_b$, \emph{i.e.,} one whose current state lies in $\mathcal{X}_b$ while the preceding $h$ states do not. Condition~\eqref{subeq: decreasing} requires that, for every history in $\mathcal{X}^{h+1}$, there exists an input under which $\mathcal{B}$ decays at rate $\lambda$, up to an increase of $\gamma\|w\|^2$ due to the disturbance. Finally, \eqref{eq: climit} ensures that this increase, bounded by $\gamma\delta$, does not exceed the decay margin $(1-\lambda)\beta$, keeping $\mathcal{B}$ strictly below $\beta$ so that $\mathcal{X}_b$ is never entered.}
	The next theorem leverages the RK-CBC in Definition~\ref{def: CBC} and ensures the safety of the dt-IAUPS-td in the presence of time-invariant delays and bounded disturbances.
	
	\begin{theorem}[\textbf{Robust Safety Guarantee}]\label{thm: model-based}
		Given a dt-IAUPS-td, let $\mathcal{B}$ be an RK-CBC for $\Upsilon$ as defined in Definition~\ref{def: CBC}, with $ \gamma \delta$ satisfying \eqref{eq: climit}. Then, {under the R-SC $u(\cdot)$ enforcing~\eqref{subeq: decreasing} and any disturbance $w(\cdot)$ with $w(k)\in\mathcal{W}(\delta)$,} for any initial sequence $\mathbf{x}_0 \in \mathcal{X}^{h+1}_a$ and {all $k \in \mathbb{N}$ prior to the first exit of the closed-loop trajectory from $\mathcal{X}$}, one has $\mathbf{x}_{\mathbf{x}_0uw}(k) \notin \mathcal{X}_b \times (\mathcal{X}\backslash\mathcal{X}_b)^h$, \emph{i.e.,} $\Upsilon\models\Theta$.
	\end{theorem}
	
	\begin{proof}
		By \eqref{subeq: initial}, the initial state satisfies $\mathcal{B}(\mathbf{x}_0 )\le \eta < \beta$. We show that if $\mathcal{B}(\mathbf{x}(k))<\beta$ {and $\mathbf{x}(k)\in\mathcal{X}^{h+1}$}, then $\mathcal{B}(\mathbf{x}(k+1))<\beta$. To this end, one can write the following chain of inequalities:
		\begin{align*}
			\mathcal{B}(\mathbf{x}(k+1)) &\overset{\eqref{subeq: decreasing}}{\leq} \lambda\mathcal B(\mathbf{x}(k)) + \gamma \|w\|^2 \overset{\eqref{eq:Wdelta}}{<}\lambda\beta + \gamma\delta \\
			&\overset{\eqref{eq: climit}}{{\le}}\lambda\beta + (1-\lambda)\beta = \beta.
		\end{align*}
		Hence, according to \eqref{subeq: unsafe}, $\mathbf{x}(k+1)\notin \mathcal{X}_b \times (\mathcal{X}\backslash\mathcal{X}_b)^h$. Consequently, {by induction over the time steps at which $\mathbf{x}(k)\in\mathcal{X}^{h+1}$,} all system {trajectories $\mathbf{x}_{\mathbf{x}_0uw}(k)$} will remain outside $\mathcal X_b \times (\mathcal X \setminus \mathcal X_b)^h$ {prior to their first exit from $\mathcal{X}$}, thus completing the proof.
	\end{proof} 
	
	{\begin{remark}\label{rem:LK}
			Although $\mathcal{B}$ may resemble a Lyapunov--Krasovskii functional and~\eqref{subeq: decreasing} imposes its decay along trajectories, Theorem~\ref{thm: model-based} certifies avoidance of the prescribed set $\mathcal{X}_b$, not convergence to an equilibrium: the decay in~\eqref{subeq: decreasing} serves, together with~\eqref{eq: climit}, only to keep $\mathcal{B}(\mathbf{x}(k))<\beta$ for all $k \in \mathbb{N}$ prior to exiting $\mathcal{X}$, so that~\eqref{subeq: unsafe} prevents any first entry into $\mathcal{X}_b$, while~\eqref{subeq: initial}--\eqref{subeq: unsafe} encode $\Theta$ directly into the certificate. No stability property is required or concluded, and neither Definition~\ref{def: CBC} nor Theorem~\ref{thm: model-based} refers to an equilibrium.
	\end{remark}}
	
	While the RK-CBC in Definition~\ref{def: CBC} guarantees robust safety  according to Theorem~\ref{thm: model-based}, its direct synthesis is impractical as the model-dependent terms appearing on the left-hand side of~\eqref{subeq: decreasing} are unknown. In particular, the state {evolution $x(k+1) \!=\! A_1  \mathcal M (x(k)) + A_2  \mathcal M (x(k-h)) + B \mathcal G (x(k),x(k-h)){u}(k) + w(k)$} depends on unknown system and input matrices, as well as unknown disturbances. Motivated by this key challenge, we propose our data-driven approach in the next section.
	
	\section{Data-Driven Construction of RK-CBC and R-SC}\label{sec: data scheme}
	In this section, we present our data-driven procedure to construct an RK-CBC in the form of 
	\begin{equation}\label{RK-CBC}
		\mathcal{B}(\mathbf{x})=x^\top P x + \kappa \sum_{i=1}^{h} \lambda ^i x^\top_{i} P x_{i},
	\end{equation}
	where $\lambda,\kappa \in (0,1)$, with $\lambda^{i}$ denoting the $i$-th power of $\lambda$, and
	$P \succ 0$ is a positive-definite matrix. The corresponding R-SC for the unknown dt-IAUPS-td is derived alongside this certificate. Unlike delay-free CBC, which depend only on the current state, this formulation explicitly incorporates delayed states through a summation term, allowing their influence to be directly captured in the safety synthesis process. Such functionals are closely related to those used in Lyapunov–Krasovskii theory for the stability analysis of time-delay systems~\cite{Papachristodoulou_time_delay}. The shared matrix $P$ and the geometric weights $\kappa\lambda^i$ are chosen so that, in the one-step difference $\mathcal{B}(\mathbf{x}(k+1))-\lambda\mathcal{B}(\mathbf{x}(k))$, the weighted delayed terms at consecutive times cancel pairwise, leaving only $x(k+1)$, $x(k)$, and $x(k-h)$ (see the proof of Theorem~\ref{thm: main}). This makes the synthesis conditions and the size of the resulting SOS program independent of the delay length $h$.
	
	Given the initial state history $\mathbf{x}_0=(x(0), \dots, x(-h)) \in \mathcal{X}_a^{h+1}$, we collect input-state measurements over the horizon $[0,\mathtt{T}]$, where $\mathtt{T} \in \mathbb N^{+}$ is the number of samples:
	\begin{align}\notag
			{\mathit{U}_{-}} &\!=\! [u(0)\;\; u(1)\;\; \dots \;\; u({\mathtt{T}-1})], \\\notag
			{\mathit{X}_{-}} &\!=\! [x(0)\;\; x(1)\;\; \dots \;\; x(\mathtt{T}-1)],   \\ \notag
			{\mathit{X}_{+}} &\!=\! [x(1)\;\; x(2)\;\; \dots \;\; x(\mathtt{T})], 
			\\ \notag
			{\mathit{X}}_h &\!=\! [x(-h)\;\; x(-h+1)\;\; \dots \;\; x(-h+\mathtt{T}-1)], \\\label{eq: datarep1}
			{\mathit{W}_{-}} &\!=\! [w(0)\;\;  w(1)\;\; \dots\;\; w(\mathtt{T}-1)].
	\end{align}
	The state trajectory $\mathit{X}_{+}$ records the evolution of the unknown system dynamics over discrete time steps. In particular, the data collection process begins by applying \emph{arbitrary} input sequences, which are {stored in the matrix $\mathit{U}_{-}$}. Starting from the initial state history $\mathbf{x}_0$ and given the time-invariant delay $h$, we first collect the state sequence $\mathit{X}_{-}$ and then construct the delayed state trajectory $\mathit{X}_{h}$ by appropriately shifting the state data to capture the influence of delayed states.  We emphasize that the disturbance trajectory $\mathit{W}_{-}$ is not directly measured and remains unknown.
	
	Based on the dictionary functions $\mathcal{M}(\cdot)$ and $\mathcal{G}(\cdot)$, we now construct the auxiliary trajectories required for the subsequent analysis as follows:
	\begin{align}\notag
		{\mathit{M}_{-}} &\!=\!
		\big[\mathcal{M}(x(0)), \mathcal{M}(x(1)),\dots,
		\mathcal{M}(x(\mathtt{T}-1))\big]\!,  \\\notag
		{\mathit{M}}_h &\!=\!
		\big[\mathcal{M}(x(-h)), \mathcal{M}(x(-h+1)),\dots,
		\mathcal{M}(x(-h+\mathtt{T}-1))\big]\!,\\\notag 
		{{\mathit{G}}} &\!=\!
		{\big[\mathcal{G}(x(0),x(-h)){u}(0), \mathcal{G}(x(1),x(-h+1)){u}(1), \dots,\,}\\\label{eq: datarep2}&
		\quad\quad \quad \quad \quad \quad \quad \quad \mathcal{G}(x(\mathtt{T}-1),x(-h+\mathtt{T}-1)){u}(\mathtt{T}-1)\big]\!.
	\end{align}
	
	To propose the required data-driven safety conditions for the construction of RK-CBC and its R-SC, we first form a parameterization of dt-IAUPS-td, as detailed in the following lemma.
	
	\begin{lemma}[\textbf{Parameterization of dt-IAUPS-td}]\label{lemma1}
		Let $\mathds{L}(x)\in \mathbb{R}^{M \times n}$ be a polynomial transformation matrix fulfilling
		\begin{subequations}
			\begin{align}\label{transform}
				\mathcal{M}(x)&=\mathds{L}(x)x,
			\end{align}
			for any $x \in \mathcal{X}$. By designing {the state-feedback controller in \eqref{feedback}, rewritten as}
			\begin{align}\label{cont}
				u = \mathcal{F}_1(x,x_h)x + \mathcal{F}_2(x,x_h)x_h,
			\end{align}
			with $\mathcal{F}_1, \mathcal{F}_2$ being {some matrix polynomials of $(x,x_h)$}, the system in~\eqref{eq: dt-NPS} can be equivalently represented as
			\begin{equation}\label{lem}
				x^+= \CLbr x+  \DLbr x_h + w,
			\end{equation}
			where $x^+ := x(k+1)$, and
			{\begin{align*}
					\mathcal{{C}}_1 &= [A_1\quad B],  \quad \mathcal{R}_1(x,x_h) = \begin{bmatrix} 
						\mathds{L}(x)\\
						\mathcal{G}(x,x_h)\mathcal{F}_1(x,x_h)\end{bmatrix}\!\!,\\
					\mathcal{{C}}_2 &= [A_2\quad B],  \quad \mathcal{R}_2(x,x_h) = \begin{bmatrix} 
						\mathds{L}(x_h)\\
						\mathcal{G}(x,x_h)\mathcal{F}_2(x,x_h)\end{bmatrix}\!\!.
			\end{align*}}
		\end{subequations}
	\end{lemma}
	\vspace{0.1cm}
	\begin{proof}
		{By designing the state feedback \( u\) as in~\eqref{cont}, one has}
		\begin{align*}\notag
				&A_1 \mathcal{M}(x) \!+\! A_2 \mathcal{M}(x_h) \!+\! B\mathcal{G}(x,x_h){u}\\\notag &\overset{\eqref{transform}}{=} (A_1\mathds{L}(x)  \!+\! B\mathcal{G}(x,x_h)\mathcal{F}_1(x,x_h))x \!\\\notag & ~~~+\! (A_2\mathds{L}(x_h)  \!+\! B\mathcal{G}(x,x_h)\mathcal{F}_2(x,x_h))x_h \\\notag
				&= \underbrace{[A_1\quad B]}_{\mathcal{{C}}_1}\underbrace{\begin{bmatrix} \mathds{L}(x) \\
						\mathcal{G}(x,x_h)\mathcal{F}_1(x,x_h)\end{bmatrix}}_{\mathcal{R}_1(x,x_h)}\!x \\
				& ~~~+\underbrace{[A_2\quad B]}_{\mathcal{{C}}_2}\underbrace{\begin{bmatrix} \mathds{L}(x_h) \\
						\mathcal{G}(x,x_h)\mathcal{F}_2(x,x_h)\end{bmatrix}}_{\mathcal{R}_2(x,x_h)}\! x_h.
		\end{align*}
		Then, we have
		{\begin{align}\notag
				x^+ &= A_1 \mathcal{M}(x) \!+\! A_2 \mathcal{M}(x_h) \!+\! B\mathcal{G}(x,x_h){u} \!+\! w\\\notag & = \CLbr x+  \DLbr x_h + w,
		\end{align}}
		which concludes the proof.  
	\end{proof}
	\begin{remark} 
		Condition~\eqref{transform} ensures that all subsequent expressions are formulated in terms of the state history \((x,x_1,\cdots, x_h)\), rather than the lifted variables \((\mathcal{M}(x),\mathcal{M}(x_1),\cdots,\mathcal{M}(x_h))\). This consistency is essential, as the proposed RK-CBC in~\eqref{RK-CBC} explicitly depends on \((x,x_1,\cdots, x_h)\). Representing all expressions in this form simplifies the overall formulation and facilitates the derivation of the main results in {Section~\ref{sec: data scheme}}. Moreover, without loss of generality, for any $\mathcal{M}(x)$ there always exists $\mathds{L}(x)$ fulfilling~\eqref{transform} since $\mathcal{M}(\mathbf{0}_n) = \mathbf{0}_M$. 
	\end{remark}
	
	Building on the results of Lemma~\ref{lemma1}, we propose the following theorem as the central contribution of this work, which enables the design of an RK-CBC and its R-SC using the collected trajectories in~\eqref{eq: datarep1}. 
	\begin{figure*}
		\begin{align}\label{T2}
			\widetilde{\mathcal{H}} &=
			{\begin{bmatrix}
					\displaystyle\frac{1+\mu_2}{\Delta} P^{-1} 
					& \mathbf{0}_{n\times\varphi} 
					& \mathbf{0}_{n\times\varphi} 
					& \displaystyle-\frac{1}{\Delta} P^{-1} 
					& \mathbf{0}_{n\times\varphi} 
					& \mathbf{0}_{n\times\varphi} 
					& \mathbf{0}_{n\times n} 
					& \mathbf{0}_{n\times n} \\
					\star
					& \mathbf{0}_{\varphi\times\varphi} 
					& \mathbf{0}_{\varphi\times\varphi} 
					& \mathbf{0}_{\varphi\times n} 
					& \mathbf{0}_{\varphi\times\varphi} 
					& \mathbf{0}_{\varphi\times\varphi} 
					& \widetilde{\mathcal{Z}}_1 
					& \mathbf{0}_{\varphi\times n} \\
					\star
					& \star
					& \mathbf{0}_{\varphi\times\varphi} 
					& \mathbf{0}_{\varphi\times n} 
					& \mathbf{0}_{\varphi\times\varphi} 
					& \mathbf{0}_{\varphi\times\varphi} 
					& \mathbf{0}_{\varphi\times n} 
					& \mathbf{0}_{\varphi\times n} \\
					\star
					& \star
					& \star
					& \displaystyle\frac{1+\mu_1}{\Delta} P^{-1} 
					& \mathbf{0}_{n\times\varphi} 
					& \mathbf{0}_{n\times\varphi} 
					& \mathbf{0}_{n\times n} 
					& \mathbf{0}_{n\times n} \\
					\star
					& \star
					& \star
					& \star
					& \mathbf{0}_{\varphi\times\varphi} 
					& \mathbf{0}_{\varphi\times\varphi} 
					& \mathbf{0}_{\varphi\times n} 
					& \mathbf{0}_{\varphi\times n} \\
					\star
					& \star
					& \star
					& \star
					& \star
					& \mathbf{0}_{\varphi\times\varphi} 
					& \mathbf{0}_{\varphi\times n} 
					& \widetilde{\mathcal{Z}}_2 \\
					\star
					& \star
					& \star
					& \star
					& \star
					& \star
					& \phi_1 P^{-1} 
					& \mathbf{0}_{n\times n} \\
					\star
					& \star
					& \star
					& \star
					& \star
					& \star
					& \star
					& \phi_2 P^{-1}
			\end{bmatrix}}
			\!,\quad 
			\widetilde{\mathcal{S}} =
			\begin{bmatrix}
				\mathcal{S}            & \mathbf{0}_{q\times q} & \mathbf{0}_{q\times n} & \mathbf{0}_{q\times n} \\
				\star                  & \mathcal{S}            & \mathbf{0}_{q\times n} & \mathbf{0}_{q\times n} \\
				\star                  & \star                  & \mathbf{0}_{n\times n} & \mathbf{0}_{n\times n} \\
				\star                  & \star                  & \star                  & \mathbf{0}_{n\times n}
			\end{bmatrix}
		\end{align}
		\rule{\textwidth}{0.1pt}
	\end{figure*}
	
	\begin{theorem}[\textbf{Data-driven RK-CBC and R-SC}]\label{thm: main}
		Consider a {(partially) unknown} {dt-IAUPS-td $\Upsilon = \big(A_1, A_2, B, \mathcal G, \mathcal M, \mathcal X, \mathcal U, \allowbreak \mathcal W, h\big)\!$,} with time-invariant delay $h \in \mathbb{N}^+$ and its parameterization in Lemma~\ref{lemma1}. {Given the collected data in~\eqref{eq: datarep1} and the dictionary trajectories in~\eqref{eq: datarep2}, let the data matrices be constructed as
			\begin{subequations}\label{data-mat}
				\begin{align}\label{Xi}
					\Xi_1 &= \begin{bmatrix}\mathit{M_{-}} \\[1mm] \dfrac{\mathit{G}}{2}\end{bmatrix}\!\!,\quad
					\Xi_2 = \begin{bmatrix}\mathit{M}_h \\[1mm] \dfrac{\mathit{G}}{2}\end{bmatrix}\!\!,\\\label{S-mat}
					\mathcal{S}&=\begin{bmatrix}
						X_{+} X_{+}^{\top}-\mathtt{T} \delta \mathds{I}_n & - X_{+} \Xi_1^{\top} & -X_{+} \Xi_2^{\top} \\
						\star & \Xi_1 \Xi_1^{\top} &  \Xi_1 \Xi_2^{\top} \\
						\star & \star & \Xi_2 \Xi_2^{\top}
					\end{bmatrix}\!\!.
				\end{align}
		\end{subequations}}
		Suppose there exist a constant matrix $P \succ0$, state-dependent {matrices $\widetilde{\mathcal{F}}_1(x,x_h)$, $\widetilde{\mathcal{F}}_2(x,x_h)$, constants} $\mu_1, \mu_2, \eta, \beta \in \mathbb{R}^+$ with $\beta > \eta $, $\Delta = (1+\mu_1)(1+\mu_2)-1$, $\kappa, \lambda \in (0,1)$, and $\alpha: \mathcal{X}^2\!\to\!\mathbb{R}^+$ such that {$\gamma=(1+\frac{1}{\mu_1}+\frac{1}{\mu_2})\|\sqrt{P}\|^2$ satisfies~\eqref{eq: climit} and}
		\begin{subequations}\label{Th: cons}
			\begin{align}\label{Th:con1}
				&x^\top P x \leq {\dfrac{\eta}{1+\kappa(\dfrac{\lambda-\lambda^{h+1}}{1-\lambda})}}, \quad\quad\, \forall x  \in \mathcal{X}_a,\\\label{Th:con2}
				&x^\top P x \geq \beta, \quad\quad\quad\quad\quad\quad\quad\quad\quad \forall x  \in \mathcal{X}_b,\\ \label{Th:con3}
				&\widetilde{\mathcal{H}}  + \alpha(x,x_h)\,\widetilde{\mathcal{S}} \succeq 0, \quad\quad\quad\quad\quad\,\, \forall (x, x_{h}) \in {\mathcal{X}}^{2},
			\end{align}
		\end{subequations}
		where $\widetilde{\mathcal{H}}$ and $\widetilde{\mathcal{S}}$ are as in \eqref{T2}, and
		\begin{align}\label{Ztilde}
				\widetilde{{\mathcal{Z}}}_1 = \begin{bmatrix} 
					\mathds{L}(x) P^{-1}\\
					\mathcal{G}(x,x_h)\widetilde{\mathcal{F}}_1(x,x_h)\end{bmatrix}\!\!,\quad \widetilde{\mathcal{Z}}_2 = \begin{bmatrix} 
					\mathds{L}(x_h) P^{-1}\\
					\mathcal{G}(x,x_h)\widetilde{\mathcal{F}}_2(x,x_h)\end{bmatrix}\!\!,
		\end{align}
		$\varphi = M+N$, $q = n+2(N+M)$, $\phi_1 = \lambda(1-\kappa) $, and $\phi_2 =\kappa \lambda^{h+1}$.
		Then, $\mathcal{B}(\mathbf{x})= x^\top P x + \kappa \sum_{i=1}^{h} \lambda^i x^\top_{i} P x_{i}$ is an RK-CBC, {and $u = \mathcal{F}_1(x,x_h) x +\mathcal{F}_2(x,x_h) x_h$ with $\mathcal{F}_1(x,x_h)=\widetilde{\mathcal{F}}_1(x,x_h) P$ and $\mathcal{F}_2(x,x_h)=\widetilde{\mathcal{F}}_2(x,x_h) P$ is its corresponding R-SC.}
	\end{theorem}
	
	\begin{proof}
		We begin by demonstrating that the conditions in \eqref{Th:con1} and \eqref{Th:con2} imply the satisfaction of \eqref{subeq: initial} and \eqref{subeq: unsafe}, respectively. By selecting an RK-CBC of the form $\mathcal{B}(\mathbf{x})= x^\top P x + \kappa \sum_{i=1}^{h} \lambda^i x^\top_{i} P x_{i}$, one has
		\begin{align*}
			\sup_{\mathbf{x} \in \mathcal{X}_a^{h+1}} \mathcal{B}(\mathbf{x}) 
			&\stackrel{\kappa, \lambda > 0}{=} \sup_{x \in \mathcal{X}_a} x^\top P x + \kappa \sum_{i=1}^{h} \lambda^i \sup_{x_i \in \mathcal{X}_a} x_i^\top P x_i \\
			&~~~= \big(1 + \kappa \sum_{i=1}^{h} \lambda^i \big) \sup_{x \in \mathcal{X}_a} x^\top P x \\
			&~~~= {\big(1 + \kappa ( \frac{\lambda - \lambda^{h+1}}{1 - \lambda}) \big) \sup_{x \in \mathcal{X}_a} x^\top P x.}
		\end{align*}
		Hence, condition \eqref{Th:con1} is equivalent to \eqref{subeq: initial}. Furthermore, we have
		\begin{align*}
			\inf_{\mathbf{x} \in \mathcal{X}_b \times (\mathcal{X} \backslash \mathcal{X}_b)^h} \!\mathcal{B}(\mathbf{x}) 
			&\stackrel{\kappa, \lambda > 0}{=} \!\!\inf_{x \in \mathcal{X}_b} x^\top P x \!+\! \kappa \sum_{i=1}^{h} \! \lambda^i \!\!\! \inf_{x_i \in (\mathcal{X} \backslash \mathcal{X}_b)} \!\! x_i^\top P x_i \\
			&~~\stackrel{P \succ 0}{\geq} \inf_{x \in \mathcal{X}_b} x^\top P x.
		\end{align*}
		Thus, satisfying condition \eqref{Th:con2} ensures that \eqref{subeq: unsafe} is also fulfilled.
		We now demonstrate that condition \eqref{Th:con3} fulfills \eqref{subeq: decreasing}, as well. To do so, we have
		\begin{align}\notag
			\mathcal B(\mathbf x(k+1))
			&= x^\top(k+1) P x(k+1)\\\notag &~~~+ \kappa \sum_{i=1}^{h} \lambda^i x^\top(k-i+1) P x(k-i+1)\\\notag
			&= x^\top(k+1) P x(k+1) + \kappa\lambda x^\top(k) P x(k)\\\label{telescopic3} &~~~
			+ \kappa\sum_{i=2}^{h} \lambda^{i} x^\top(k-i+1) P x(k-i+1) .
		\end{align}
		Given $\mathbf{x}(k)$ in \eqref{history}, by subtracting $\lambda \mathcal{B}(\mathbf{x}(k))$ from both sides of~\eqref{telescopic3} and re-indexing the summation over delayed states, we have
		\begin{align}\notag
			&\mathcal B(\mathbf{x}(k+1)) - \lambda\mathcal B(\mathbf{x}(k))
			\\\notag & = x^\top(k+1) P x(k+1) + \kappa \lambda x^\top(k) P x(k)\\\notag &  + \kappa \sum_{i=1}^{h-1} \lambda^{i+1} x^\top(k-i) P x(k-i)
			\\\notag &~~~\overbrace{- \lambda x^\top(k) P x(k) - \kappa \sum_{i=1}^{h} \lambda^{i+1} x^\top(k-i) P x(k-i)}^{\lambda\mathcal B(\mathbf{x}(k))}.
		\end{align}
		By cancellation of the overlapping terms in the sums over \( P\), we simplify the series and obtain
		\vspace{-1mm}
		\begin{align}\notag
			&\mathcal B(\mathbf{x}(k+1)) - \lambda\mathcal B(\mathbf{x}(k))\\\notag
			&= x^{\top}(k+1)  P x(k+1)  -\lambda x^{\top}(k) P x(k)+\kappa\lambda x^{\top}(k) P x(k) \\\notag &~~~ -\kappa \lambda^{h+1} x^{\top}(k-h) P x(k-h)\\\notag
			&= x^{\top}(k+1) P x(k+1) +\lambda(\kappa-1) x^{\top}(k) P x(k)\\\label{Telescopic} &~~~-\kappa \lambda^{h+1} x^{\top}(k-h) P x(k-h).
		\end{align}
		From \eqref{history-simple}, we set \( x(k-h) = x_h \) as the state delayed by \( h \) steps and \( x(k) = x \) as the current state. Substituting the system dynamics in~\eqref{eq: dt-NPS} into the expression in~\eqref{Telescopic} yields
		\vspace{-1mm}
		\begin{align}\notag
			&\mathcal B(\mathbf{x}^+ ) - \lambda\mathcal B(\mathbf{x}) \\\notag
			&= {\big(A_1\mathcal{M}(x) + A_2\mathcal{M}(x_h) 
				+ B\mathcal{G}(x,x_h){u}+ w\big)^\top} \\\notag
			&~~~ {P\big(A_1\mathcal{M}(x) + A_2\mathcal{M}(x_h) 
				+ B\mathcal{G}(x,x_h){u}+ w\big)} \\\notag &~~~ +\lambda(\kappa-1) x^{\top} P x-\kappa \lambda^{h+1} x_h^{\top} P x_h  \\\notag
			&\overset{\eqref{lem}}{=}
			\big(\CLbr x+  \DLbr x_h + w\big)^\top \\\notag
			&~~~ P\big(\CLbr x+  \DLbr x_h + w\big) \\\notag &~~~ +\lambda(\kappa-1) x^{\top} P x -\kappa \lambda^{h+1} x_h^{\top} P x_h\\\notag
			&= x^\top(\CLbr)^\top P(\CLbr)x \\\notag
			&\quad + x_h^\top\!\big(\DLbr \big)^\top
			P\big(\DLbr \big)x_h \\\notag
			&\quad + 2\,x^\top(\CLbr)^\top P\big(\DLbr \big)x_h  \\\notag
			&\quad + 2\,\underbrace{x^\top(\CLbr)^\top \sqrt{P}}_{a}\,
			\underbrace{\sqrt{P}\,w}_{b}  \\\notag
			&\quad + 2\,\underbrace{x_h^\top\!\big(\DLbr \big)^\top \sqrt{P}}_{c} \underbrace{\sqrt{P} w}_{d}
			+ w^\top P w \\\notag &~~~+\lambda(\kappa-1) x^{\top} P x-\kappa \lambda^{h+1} x_h^{\top} P x_h.
		\end{align}
		\begin{figure*}
			\begin{align}\label{final}
				{\Gamma = \begin{bmatrix}
						\dfrac{1+\mu_2}{\Delta}\,P^{-1} - \dfrac{1}{\lambda(1-\kappa)}JP^{-1} J^\top & -\,\dfrac{1}{\Delta}P^{-1} \\[8pt]
						\star& \dfrac{1+\mu_1}{\Delta}\,P^{-1} - \dfrac{1}{\kappa \lambda^{h+1} }Y P^{-1} Y^\top
				\end{bmatrix}}
			\end{align}
			\rule{\textwidth}{0.1pt}
		\end{figure*}
		According to the Cauchy-Schwarz inequality~\cite{bhatia1995cauchy}, \emph{i.e.,} $a b \leq \|a\|\|b\|$, for any $a^{\top}, b \in \mathbb{R}^n$, followed by employing Young's inequality~\cite{young1912classes}, \emph{i.e.,} $\|a\|\|b\| \leq \frac{\mu_1}{2}\|a\|^2+\frac{1}{2 \mu_1}\|b\|^2$, for any $\mu_1 \in \mathbb{R}^{+}$,
		similarly $\|c\|\|d\| \leq \frac{\mu_2}{2}\|c\|^2+\frac{1}{2 \mu_2}\|d\|^2$ for any $\mu_2 \in \mathbb{R}^{+}$, one has
		\begin{align}\notag
			&\mathcal B(\mathbf{x}^+ ) - \lambda\mathcal B(\mathbf{x}) \\\notag &\leq x^\top\big(\CLbr \big)^{\top} P  \big(\CLbr\big)x \\\notag
			&~~~ + x_h^\top\!\big(\DLbr \big)^\top
			P\big(\DLbr \big)x_h \\\notag
			&~~~ + \mu_1 x^\top\big(\CLbr \big)^{\top} P  \big(\CLbr\big)x \\\notag
			&~~~+\mu_2 x_h^\top\!\big(\DLbr \big)^\top
			P\big(\DLbr \big)x_h  \\\notag
			&~~~+ 2\,x^\top(\CLbr)^\top P\big(\DLbr \big)x_h  \\\notag
			& ~~~ + \frac{1}{\mu_2}\Vert\sqrt{P}\Vert^2\Vert w \Vert^2+ \frac{1}{\mu_1}\Vert\sqrt{P}\Vert^2\Vert w \Vert^2 +  \Vert\sqrt{P}\Vert^2\Vert w \Vert^2\\\notag
			&~~~ +\lambda(\kappa-1) x^{\top} P x-\kappa \lambda^{h+1} x_h^{\top} P x_h.
		\end{align}
		Expanding the quadratic forms, collecting terms in  $x$ and  $x_{h}$, we obtain 
		\begin{align}\notag
			&\mathcal B(\mathbf{x}^+ ) - \lambda\mathcal B(\mathbf{x}) \\\notag &\leq
			x^\top\big(\big(1 \!+\! \mu_1)(\CLbr \big)^{\top} \!\! P  \big(\CLbr\big)\\\notag &~~~+\! \lambda(\kappa-1)P \big)x \!+ 2\,x^\top(\CLbr)^\top \! P\big(\DLbr \big)x_h  \\\notag
			&~~~+x_h^\top\!\big((1 \!+\! \mu_2) \big(\DLbr \big)^\top
			\!\! P\big(\DLbr \big) \\\label{new87}
			& ~~~ - \kappa \lambda^{h+1}P \big)x_h  + \underbrace{(1 \!+\! \frac{1}{\mu_1} \!+\! \frac{1}{\mu_2})\Vert\sqrt{P}\Vert^2}_\gamma\Vert w \Vert^2.
		\end{align}
		Then, one can reformulate \eqref{new87} as 
		\begin{align}\notag
			\mathcal B(\mathbf{x}^+ ) - \lambda \mathcal B(\mathbf{x})  \le
			\begin{bmatrix} x \\ x_h \end{bmatrix}^{\!\top}
			\underbrace{\begin{bmatrix} \Lambda_{11} & \Lambda_{12} \\ \star & \Lambda_{22} \end{bmatrix}}_{\Lambda}
			\begin{bmatrix} x \\ x_h \end{bmatrix}
			+ \gamma  \|w\|^{2},
		\end{align}
		with
		\begin{align*}
			\Lambda_{11} &= (1 \!+\! \mu_1)\big(\CLbr \big)^{\top}\! P  \big(\CLbr\big) \!+\! \lambda(\kappa-1) P \!\\
			\Lambda_{12} &= (\CLbr)^\top P\big(\DLbr \big), \\
			\Lambda_{21} &= \Lambda^\top_{12},\\
			\Lambda_{22} &= (1+\mu_2) \big(\DLbr \big)^\top
			P\big(\DLbr \big) - \kappa \lambda^{h+1}  P.
		\end{align*}
		It is clear that if $\Lambda  \preceq 0$, then one can conclude that $\mathcal B(\mathbf{x}^+ ) - \lambda\mathcal B(\mathbf{x}) \leq \gamma \|w\|^{2}$ with $\gamma = (1 \!+\! \frac{1}{\mu_1} \!+\! \frac{1}{\mu_2})\Vert\sqrt{P}\Vert^2$. To streamline the argument, we set
		\begin{align*}
			J &= \CLbr,\quad 
			Y= \DLbr, \\[2pt]
			Z &= \mathsf{blkdiag}(J,	Y), \quad
			Q = \begin{bmatrix}(1+\mu_1)P & P \\ P & (1+\mu_2)P\end{bmatrix}\!\!, \\[2pt]
			D &= \mathsf{blkdiag}(\lambda(1-\kappa)P, \kappa \lambda^{h+1} P),
		\end{align*}
		which implies $\Lambda =  Z^\top Q 	 Z - D$. Since $P \succ 0$, it is clear that $D \succ 0$. Additionally, $	Q = 	S \otimes P$ with 
		\begin{align*}
			S &= \begin{bmatrix} 1+\mu_1 & 1 \\ 1 & 1+\mu_2 \end{bmatrix}\!\!.
		\end{align*}
		Accordingly, since $\mu_1,\mu_2>0$, it is clear that $S\succ 0$, and therefore $S^{-1}$ is given by
		\begin{align*}
			S^{-1} &= \tfrac{1}{\Delta}\begin{bmatrix}1+\mu_2 & -1 \\ -1 & 1+\mu_1\end{bmatrix}\!\!,
		\end{align*}
		with $\Delta = (1+\mu_1)(1+\mu_2)-1$. Consequently, $Q\succ 0$ and is therefore invertible, with its inverse given by
		\begin{align*}
			Q^{-1} &= S^{-1} \otimes P^{-1} = \tfrac{1}{\Delta}\begin{bmatrix}
				(1+\mu_2)P^{-1} & -P^{-1} \\
				-\,P^{-1} & (1+\mu_1)P^{-1}
			\end{bmatrix}\!\!.
		\end{align*}
		By utilizing the Schur complement \cite{zhang2006schur}, the following equivalences hold:
		\begin{align*}
			\Lambda \preceq 0
			&\;\;\Longleftrightarrow\;\; 	D -  Z^\top Q Z \succeq 0 \\[2pt]
			&\;\;\Longleftrightarrow\;\; 
			\begin{bmatrix}
				D & 	Z^\top \\
				Z & 	Q^{-1}
			\end{bmatrix} \succeq 0 \\[2pt]
			&\;\;\Longleftrightarrow\;\; 	\Gamma = Q^{-1} - Z D^{-1} Z^\top \succeq 0,
		\end{align*}
		where 
		\begin{align*}
			Z 	D^{-1} Z^\top &= 
			\begin{bmatrix}
				\dfrac{1}{\lambda(1-\kappa)} J P^{-1}J^\top &  \mathbf{0}_{n \times n} \\
				\mathbf{0}_{n \times n}  & \dfrac{1}{\kappa \lambda^{h+1} } Y P^{-1}Y^\top
			\end{bmatrix}\!\!.
		\end{align*}
		Hence, $\Gamma = Q^{-1} - Z 	D^{-1} 	Z^\top$ can be constructed as in~\eqref{final}. Accordingly, $\Gamma\succeq 0$ in~\eqref{final} can be rewritten in the quadratic matrix form as
		\begin{align} \label{eq: Gamma}
			\Gamma = \xi^\top \mathcal{H} \,\xi \succeq 0,
		\end{align}
		with
		\begin{align}\notag
			\xi&=\begin{bmatrix}
				\mathds{I}_n & \mathbf{0}_{n \times n} \\
				\mathbf{0}_{n \times n} & \mathds{I}_n  \\
				\mathcal{{C}}^\top_1  & \mathbf{0}_{\varphi  \times n} \\
				\mathbf{0}_{\varphi  \times n} & \mathcal{{C}}^\top_2 
			\end{bmatrix}\!\!,
			\\\notag      
			\mathcal{H}&=\begin{bmatrix}
				\dfrac{1+\mu_2}{\Delta} P^{-1} & -\dfrac{1}{\Delta} P^{-1} & \mathbf{0}_{n \times \varphi}& \mathbf{0}_{n \times \varphi} \\
				\star & \dfrac{1+\mu_1}{\Delta} P^{-1} & \mathbf{0}_{n \times \varphi}&\mathbf{0}_{n \times \varphi}\\
				\star & 	\star  & \mathcal{Z}_1 & \mathbf{0}_{\varphi  \times \varphi}\\
				\star  & 	\star  & 	\star  & \mathcal{Z}_2 
			\end{bmatrix}\!\!,
			\\\notag
			\mathcal{Z}_1 &= -	\dfrac{1}{\lambda(1-\kappa)}  \mathcal{R}_1(x,x_h)P^{-1}\mathcal{R}^\top_1(x,x_h),
			\\\notag
			\mathcal{Z}_2 &= 	- \dfrac{1}{\kappa \lambda^{h+1} } \mathcal{R}_2(x,x_h) P^{-1} \mathcal{R}^\top_2(x,x_h),
		\end{align}
		where $\varphi = M+N$.
		\begin{figure*}
			{\begin{align}\label{T1}
					K^\top \mathcal{H} K =
					\begin{bmatrix}
						\displaystyle\frac{1+\mu_2}{\Delta} P^{-1} 
						& \mathbf{0}_{n\times\varphi} 
						& \mathbf{0}_{n\times\varphi} 
						& \displaystyle-\frac{1}{\Delta} P^{-1} 
						& \mathbf{0}_{n\times\varphi} 
						& \mathbf{0}_{n\times\varphi} \\[6pt]
						\star
						& \mathcal{Z}_1 
						& \mathbf{0}_{\varphi\times\varphi} 
						& \mathbf{0}_{\varphi\times n} 
						& \mathbf{0}_{\varphi\times\varphi} 
						& \mathbf{0}_{\varphi\times\varphi} \\
						\star
						& \star
						& \mathbf{0}_{\varphi\times\varphi} 
						& \mathbf{0}_{\varphi\times n} 
						& \mathbf{0}_{\varphi\times\varphi} 
						& \mathbf{0}_{\varphi\times\varphi} \\[3pt]
						\star
						& \star
						& \star
						& \displaystyle\frac{1+\mu_1}{\Delta} P^{-1} 
						& \mathbf{0}_{n\times\varphi} 
						& \mathbf{0}_{n\times\varphi} \\[6pt]
						\star
						& \star
						& \star
						& \star
						& \mathbf{0}_{\varphi\times\varphi} 
						& \mathbf{0}_{\varphi\times\varphi} \\
						\star
						& \star
						& \star
						& \star
						& \star
						& \mathcal{Z}_2
					\end{bmatrix}
			\end{align}}
			\rule{\textwidth}{0.1pt}
		\end{figure*}
		
		The main challenge in satisfying \eqref{eq: Gamma} is that the matrices $\mathcal{{C}}_1$ and $\mathcal{{C}}_2$ are unknown. However, by utilizing the collected data sets in~\eqref{eq: datarep1} and \eqref{eq: datarep2}, one has
		\begin{align}\label{new35}
			\mathit{X}_{+} = \mathcal{{C}}_1 \Xi_1 + \mathcal{{C}}_2 \Xi_2 + \mathit{W}_{-},
		\end{align}
		with {$\Xi_1,\Xi_2$ as in~\eqref{Xi}}. On the other hand, by applying the Schur complement \cite{zhang2006schur} to  \eqref{eq:Wdelta}, one can verify that {$w w^\top \preceq \delta \mathds{I}_n$ for all $w \in \mathcal{W}(\delta)$. Since $\mathit{W}_{-} \mathit{W}^\top_{-}= \sum_{k=0}^{\mathtt{T}-1} w(k) w^\top(k)$, we have}
		\begin{align}\notag
			\mathtt{T}\delta \mathds{I}_n \succeq &~	\mathit{W_{-}} 	\mathit{W^\top_{-}}\\\notag
			\overset{\eqref{new35}}{=}  & (	\mathit{X_{+}} - \mathcal{{C}}_1 \Xi_1  - \mathcal{{C}}_2 \Xi_2  ) (	\mathit{X_{+}} - \mathcal{{C}}_1 \Xi_1  - \mathcal{{C}}_2 \Xi_2)^\top \\\notag
			=  & \mathit{X_{+}} \mathit{X^\top_{+}}  - \mathcal{{C}}_1 \Xi_1\mathit{X^\top_{+}}  - \mathcal{{C}}_2 \Xi_2 \mathit{X^\top_{+}} - \mathit{X_{+}}  \Xi^\top_1 \mathcal{{C}}^\top_1\\\notag &~~~ - \mathit{X_{+}}  \Xi^\top_2 \mathcal{{C}}^\top_2  + \mathcal{{C}}_1 \Xi_1 \Xi^\top_1 \mathcal{{C}}^\top_1 + \mathcal{{C}}_2 \Xi_2\Xi^\top_2 \mathcal{{C}}^\top_2 \\\label{eq: assum_t} &~~~ + \mathcal{{C}}_1 \Xi_1  \Xi^\top_2 \mathcal{{C}}^\top_2 + \mathcal{{C}}_2  \Xi_2 \Xi^\top_1 \mathcal{{C}}^\top_1.
		\end{align}
		{By defining $\Psi$ as the right-hand side of~\eqref{eq: assum_t} minus $\mathtt{T}\delta \mathds{I}_n$, inequality~\eqref{eq: assum_t} can be rewritten as the quadratic matrix constraint $\Psi =\aleph^\top \mathcal{S} \,\aleph \preceq 0$, with $\mathcal{S}$ as in~\eqref{S-mat} and $\aleph=[\mathds{I}_n;\,\mathcal{C}_1^{\top};\,\mathcal{C}_2^{\top}]$.}
		Then, one can conclude that
		{\begin{align} \notag
				&\begin{bmatrix}
					\Psi & \mathbf{0}_{n \times n} \\
					\mathbf{0}_{n \times n} & 	\Psi
				\end{bmatrix} \\ & = \begin{bmatrix}
					\aleph & \mathbf{0}_{q \times n}  \\
					\mathbf{0}_{q \times n}  &\aleph
				\end{bmatrix}^\top \begin{bmatrix}
					\mathcal{S} & \mathbf{0}_{q \times q}  \\
					\mathbf{0}_{q \times q}  & \mathcal{S}
				\end{bmatrix} \begin{bmatrix}
					\aleph & \mathbf{0}_{q \times n} \\
					\mathbf{0}_{q \times n}  &\aleph
				\end{bmatrix} \preceq 0, \label{eq: assum_t double}
		\end{align} }
		with $q = n+2(N+M)$. On the other hand, the relation in \eqref{eq: Gamma} can be reformulated as
		\begin{align} \label{eq: Gamma1}
			\Gamma =  \overbrace{\begin{bmatrix}
					\aleph & \mathbf{0}_{q \times n} \\
					\mathbf{0}_{q \times n} &\aleph
				\end{bmatrix}^\top\!\!\!\! K^\top}^{\xi^\top } \mathcal{H} \overbrace{K  \begin{bmatrix}
					\aleph & \mathbf{0}_{q \times n} \\
					\mathbf{0}_{q \times n} &\aleph
			\end{bmatrix}}^{\xi}\succeq 0,
		\end{align}
		with $K = [K_1\,\, K_2]$, and
		\begin{align*}
			K_1 &\!=\!
			\begin{bmatrix}
				\mathds{I}_n                & \mathbf{0}_{n \times \varphi }           & \mathbf{0}_{n\times \varphi } \\
				\mathbf{0}_{n\times n}              & \mathbf{0}_{n\times \varphi }           & \mathbf{0}_{n\times \varphi } \\
				\mathbf{0}_{\varphi \times n}          & \mathds{I}_{\varphi }            & \mathbf{0}_{\varphi \times \varphi } \\
				\mathbf{0}_{\varphi \times n}          & \mathbf{0}_{\varphi \times \varphi }       & \mathbf{0}_{\varphi \times \varphi }
			\end{bmatrix}\!\!,
			K_2 \!=\!
			\begin{bmatrix}
				\mathbf{0}_{n\times n}              & \mathbf{0}_{n\times  \varphi}           & \mathbf{0}_{n\times  \varphi} \\
				\mathds{I}_n               & \mathbf{0}_{n\times  \varphi}           & \mathbf{0}_{n\times  \varphi} \\
				\mathbf{0}_{ \varphi\times n}          & \mathbf{0}_{ \varphi \times  \varphi}       & \mathbf{0}_{ \varphi \times  \varphi} \\
				\mathbf{0}_{ \varphi\times n}          & \mathbf{0}_{ \varphi \times  \varphi}       & \mathds{I}_{ \varphi}
			\end{bmatrix}\!\!.
		\end{align*}
		Therefore, according to the S-procedure \cite{caverly2019lmi}, in order to satisfy $\Gamma \succeq 0 $ in~\eqref{eq: Gamma1} conditioned on the data conformity in \eqref{eq: assum_t double}, it is sufficient that there exists an $\alpha: \mathcal{X}^2\!\to\!\mathbb{R}^+$ such that
		{\begin{align*}
				&\Gamma + \alpha(x,x_h) \begin{bmatrix}
					\Psi & \mathbf{0}_{n \times n} \\
					\mathbf{0}_{n \times n} & 	\Psi
				\end{bmatrix}   \succeq 0\\ & \Longleftrightarrow \begin{bmatrix}
					\aleph & \mathbf{0}_{q \times n}   \\
					\mathbf{0}_{q \times n}  &\aleph
				\end{bmatrix}^\top\!\!\!\Big(K^\top \mathcal{H} K  + \alpha(x,x_h) \begin{bmatrix}
					\mathcal{S} & \mathbf{0}_{q \times q} \\
					\mathbf{0}_{q \times q} & \mathcal{S}
				\end{bmatrix} \Big)\\ & ~~~~~~\begin{bmatrix}
					\aleph & \mathbf{0}_{q \times n}   \\
					\mathbf{0}_{q \times n}   &\aleph
				\end{bmatrix} \succeq 0,
		\end{align*}}
		where $K^\top \mathcal{H} K$ is constructed as in~\eqref{T1}.
		
		\noindent The only remaining challenge is the bilinearity arising from the products {$\mathcal{F}_1(x,x_h)$ and $\mathcal{F}_2\left(x,x_h\right)$} with $P^{-1}$ in the blocks $\mathcal{Z}_1$ and $\mathcal{Z}_2$ of $\mathcal{H}$. To address this, we apply the Schur complement \cite{zhang2006schur}, which yields an equivalent formulation where the decision variables {$P^{-1}, \widetilde{\mathcal{F}}_1(x,x_h), \widetilde{\mathcal{F}}_2(x,x_h)$,}  and $ \alpha(x, x_h)$ appear linearly, as follows:
		\begin{align*}
			K^\top\! \mathcal{H} K \!+\! \alpha(x,x_h)\! \begin{bmatrix}
				\mathcal{S} \!\!&\!\! \mathbf{0}_{q \times q} \\
				\mathbf{0}_{q \times q} \!\!&\!\! \mathcal{S}
			\end{bmatrix}  \! \succeq \! 0  \Longleftrightarrow \underbrace{\widetilde{\mathcal{H}}  \!+\! \alpha(x,x_h)\,\widetilde{\mathcal{S}} \succeq 0}_{\textbf{\text{Condition}~\eqref{Th:con3}}},
		\end{align*}
		{with $\widetilde{\mathcal{Z}}_1,\widetilde{\mathcal{Z}}_2$ as in~\eqref{Ztilde},} $\phi_1 = \lambda(1-\kappa) $, $\phi_2 =\kappa \lambda^{h+1}$,  $\widetilde{\mathcal{H}}$ and $\widetilde{\mathcal{S}}$ in~\eqref{T2}.
		{Hence, {since~\eqref{eq: climit} holds with the above $\gamma$,} $\mathcal{B}(\mathbf{x})=x^\top P x + \kappa \sum_{i=1}^{h}\lambda^i x^\top_{i} P x_{i}$ is an RK-CBC and $u = \mathcal{F}_1(x,x_h) x + \mathcal{F}_2(x,x_h) x_h$ with $\mathcal{F}_1(x,x_h)=\widetilde{\mathcal{F}}_1(x,x_h) P$ and $\mathcal{F}_2(x,x_h)=\widetilde{\mathcal{F}}_2(x,x_h) P$ is its associated R-SC, thereby concluding the proof. }
	\end{proof}
	
	{\begin{remark}\label{rem:role}
				The data contribute to Theorem~\ref{thm: main} only through $\mathcal{S}$ in~\eqref{S-mat}, built from $X_+$, $M_-$, $M_h$, $G$ and $\delta$ without any knowledge of $A_1,A_2,B$. By~\eqref{eq: assum_t double}, $\mathcal{S}$ characterizes all pairs $(\mathcal{C}_1,\mathcal{C}_2)$ consistent with the measured trajectory and the bound~\eqref{eq:Wdelta} (including the true one), and condition~\eqref{Th:con3} enforces~\eqref{subeq: decreasing} for every such pair simultaneously. Moreover, no controllability or stabilizability of~\eqref{eq: dt-NPS} is required, since safety concerns avoidance of $\mathcal{X}_b$ rather than convergence. Instead, feasibility of~\eqref{Th:con3} requires a controller~\eqref{feedback} and RK-CBC~\eqref{RK-CBC} satisfying~\eqref{subeq: decreasing} for every pair in this family; if no such controller and certificate exist, \eqref{Th:con3} is infeasible, in the same way that the model-based condition~\eqref{subeq: decreasing} is infeasible for a known system admitting no RK-CBC of the form~\eqref{RK-CBC}.
	\end{remark}}
	
	{\begin{remark}\label{rem:data}
			Persistence of excitation, classically a rank condition on the data of linear systems{~\cite{WILLEMS2005325}}, has no direct counterpart for~\eqref{eq: dt-NPS}: its nonlinear extensions cover only classes admitting a linear representation~\cite{Alsalti_Persistently}. The relevant analogue here is full row rank of the dictionary data matrices $\Xi_1,\Xi_2$, which requires $\mathtt{T}\geq M+N$ and facilitates the feasibility of~\eqref{Th:con3}. Indeed, if $\Xi_1$ is row-rank deficient, the principal submatrix $\big[\begin{smallmatrix}\alpha\,\Xi_1\Xi_1^\top & \widetilde{\mathcal{Z}}_1\\ \star & \phi_1P^{-1}\end{smallmatrix}\big]\succeq 0$ of~\eqref{Th:con3} forces $\operatorname{col}(\widetilde{\mathcal{Z}}_1)\subseteq\operatorname{col}(\Xi_1)$ for all $(x,x_h)\in\mathcal{X}^2$, \emph{i.e.,} the polynomial matrix $\widetilde{\mathcal{Z}}_1$ {is confined to a fixed proper subspace, which restricts $P^{-1}$ and $\widetilde{\mathcal{F}}_1$}; the same holds for $\Xi_2$, $\widetilde{\mathcal{Z}}_2$, and $\phi_2P^{-1}$.
	\end{remark}}
	
	\begin{remark}\label{quadratic}
		We adopt a quadratic RK-CBC in~\eqref{RK-CBC} to facilitate the reformulation of condition~\eqref{subeq: decreasing} into a tractable matrix inequality (cf.~condition~\eqref{Th:con3}). In addition, we note that data collection is required to take place while the system operates {within $\mathcal{X}\setminus\mathcal{X}_b$}. This requirement is motivated by practical considerations, rather than by limitations of the underlying theoretical framework, whose results hold for any trajectory of~\eqref{eq: dt-NPS} under~\eqref{eq:Wdelta}. Specifically, entering $\mathcal{X}_b$ may damage the system, and large states inflate the dictionary data in~\eqref{eq: datarep2}. Since $\mathcal{X}_b$ is known a priori from $\Theta$, it can readily be avoided by applying small-amplitude exploratory inputs from histories well inside $\mathcal{X}\setminus\mathcal{X}_b$.
	\end{remark}
	
	\begin{remark}\label{rem:params}
		One can consider \( P^{-1} = \Omega \) in \eqref{Th:con3} and solve for \( \Omega \) so that it is symmetric and positive definite. Once \(  \Omega \) is determined, its inverse will provide the matrix \( P \). Additionally, since \( \lambda \), \( \kappa \), \( \mu_2 \) and \( \mu_1 \) in condition \eqref{Th:con3} are scalars, they are fixed a priori when solving this condition for the design of matrices {\( P^{-1} = \Omega \), \( \widetilde{\mathcal{F}}_1(x,x_h) \) and \( \widetilde{\mathcal{F}}_2(x,x_h) \).} {Their roles are as follows: $\lambda$ is the decay rate in~\eqref{subeq: decreasing} and, through~\eqref{eq: climit} and the factor in~\eqref{Th:con1}, trades the admissible disturbance bound against the gap between $\eta$ and $\beta$; $\kappa$ weights the delayed history in~\eqref{RK-CBC} and, via $\phi_1=\lambda(1-\kappa)$ and $\phi_2=\kappa\lambda^{h+1}$, balances the current- and delayed-channel blocks of~\eqref{Th:con3}; $\mu_1,\mu_2$ stem from Young's inequalities and enter~\eqref{Th:con3} through $1+\mu_i$ and $\gamma=(1+\frac{1}{\mu_1}+\frac{1}{\mu_2})\|\sqrt{P}\|^2$, so that larger values ease~\eqref{eq: climit} while making~\eqref{Th:con3} more restrictive. Thus $\gamma$ is not a design parameter but follows from $P$, $\mu_1$, $\mu_2$. In practice, the scalars are selected by a coarse grid search over their ranges, as condition~\eqref{Th:con3} is convex for each fixed choice (see Algorithm~\ref{alg}).}
	\end{remark}
	
	{\begin{remark}\label{Complex}
			While the model-based conditions in \eqref{subeq: initial}–\eqref{subeq: decreasing} enforce safety by formulating constraints over $(h+1)$-fold Cartesian products of the state, initial, and unsafe sets, the proposed data-driven conditions in Theorem~\ref{thm: main} involve at most a \emph{two-fold} Cartesian product of these sets in condition~\eqref{Th:con3}, while no Cartesian products of the initial and unsafe sets are required for conditions~\eqref{Th:con1}–\eqref{Th:con2}. This substantial reduction in set-product complexity significantly alleviates the computational burden associated with safety synthesis for time-delayed systems and constitutes a key distinguishing feature of the proposed data-driven framework. {Moreover, the delay $h$ contributes to~\eqref{Th: cons} only through the scalars $1+\kappa(\frac{\lambda-\lambda^{h+1}}{1-\lambda})$ in~\eqref{Th:con1} and $\phi_2=\kappa\lambda^{h+1}$ in~\eqref{T2} (tightening~\eqref{Th:con1} and~\eqref{Th:con3} as $h$ grows) and the entries of the delayed data $M_h$ in~\eqref{eq: datarep2}, whereas the dimensions of $\widetilde{\mathcal{H}}$ and $\widetilde{\mathcal{S}}$ in~\eqref{T2}, and hence the size of the resulting SOS program, are independent of $h$.}
	\end{remark}}
	
	{\begin{remark}\label{rem:general}
			Mixed monomials in $(x,x_h)$ can be included in~\eqref{eq: dt-NPS} by replacing $\mathcal{M}(x)$, $\mathcal{M}(x_h)$ with a joint dictionary $\overline{\mathcal{M}}(x,x_h)\in\mathbb{R}^{\overline{M}}$ and~\eqref{transform} with $\overline{\mathcal{M}}=\overline{\mathds{L}}_1(x,x_h)x+\overline{\mathds{L}}_2(x,x_h)x_h$, which exists since $\overline{\mathcal{M}}(\mathbf{0}_n,\mathbf{0}_n)=\mathbf{0}_{\overline{M}}$ (mixed input terms are already covered by $\mathcal{G}(x,x_h)$). However, this enlarges the dimension of~\eqref{Th:con3} from $4n+4(N+M)$ to approximately $4n+4(N+\overline{M})$ (with $\overline{M}>M$), increasing the number of decision variables and computational cost. Similarly, including cross terms in~\eqref{RK-CBC} via a full block matrix $\mathcal{P}\succ0$ on $\mathbb{R}^{n(h+1)}$ may reduce conservatism, but then $x_1,\ldots,x_{h-1}$ no longer cancel in~\eqref{Telescopic}, $\sup_{\mathcal{X}_a^{h+1}}\mathcal{B}$ no longer separates as in the proof of Theorem~\ref{thm: main}, and the factorization $Q=S\otimes P$ underlying~\eqref{final} is lost. Consequently,~\eqref{Th: cons} would have to be enforced over $(h+1)$-fold products in $n(h+1)$ rather than $2n$ variables, with $\frac{n(h+1)(n(h+1)+1)}{2}$ certificate variables instead of $\frac{n(n+1)}{2}$. Finally, since $\mathcal{B}$ in~\eqref{RK-CBC} is quadratic, the unknown $(\mathcal{C}_1,\mathcal{C}_2)$ enter~\eqref{subeq: decreasing} quadratically via~\eqref{lem}, yielding the matrix inequality $\Lambda\preceq0$ and enabling the S-procedure with~\eqref{eq: assum_t double}, whereas a higher-degree $\mathcal{B}$ would preclude both steps. The adopted structures in~\eqref{eq: dt-NPS} and~\eqref{RK-CBC} thus reflect a deliberate trade-off between expressiveness and tractability.
	\end{remark}}
	
	{\begin{remark}\label{rem:aug}
			Since~\eqref{eq: dt-NPS} is discrete-time, the delay can be naively removed by augmenting the state to $z=[x;x_1;\ldots;x_h]\in\mathbb{R}^{n(h+1)}$ and applying a delay-free data-driven method such as~\cite{Akbarzadeh_LCSS} (although~\cite{Akbarzadeh_LCSS} considers no disturbances). In such a case, a certificate $z^\top\mathcal{P}z$ is precisely the full-block functional of Remark~\ref{rem:general}, with the growth in set products, variables, and certificate parameters noted there. Moreover, imposing a generic disturbance bound on all components of $z$ introduces spurious uncertainty into the $nh$ exact shift rows of $z^+$, enlarging~\eqref{eq: assum_t double} with models that violate the delay structure. {The proposed framework avoids both issues:~\eqref{RK-CBC} confines the synthesis conditions to $(x,x_h)$ rather than the full history, while the two-channel parameterization in Lemma~\ref{lemma1} keeps the current and delayed components separated through $\mathcal{C}_1$ and $\mathcal{C}_2$ instead of forming an augmented delay-free system, so that no shift rows appear and the disturbance bound~\eqref{eq:Wdelta} enters the data-conformity constraint~\eqref{eq: assum_t double} in its original $n$-dimensional form.}
	\end{remark}}
	
	\subsection{Computation of RK-CBC and R-SC}\label{sec: Computation}
	{We utilize SOS programming to enforce the polynomial conditions of Theorem~\ref{thm: main}. Specifically, a polynomial $p$ is a sum of squares if $p=\sum_{j}p_j^2$ for polynomials $p_j$, and a symmetric polynomial matrix $\Pi$ is SOS if $\Pi=\mathcal{T}^\top\mathcal{T}$ for a polynomial matrix $\mathcal{T}$, implying $p\ge0$ and $\Pi\succeq0$ globally. Since $p$ is SOS if and only if $p=\nu^\top\Sigma\,\nu$ for a vector of monomials $\nu$ and some $\Sigma\succeq0$, SOS feasibility reduces to a semidefinite program (SDP)~\cite{prajna2004sostools}. Nonnegativity on $\{x\mid\mathcal{Q}(x)\ge0\}$ follows if $p-\Phi^\top\mathcal{Q}$ is SOS for a vector $\Phi$ of SOS polynomials.} To obtain the RK-CBC and its associated R-SC, we first present Lemma \ref{SOS}, which reformulates conditions \eqref{Th:con1}-\eqref{Th:con3} as an SOS {feasibility} problem. Finally, we provide Algorithm \ref{alg}, which details the required steps to obtain the proposed data-driven results. 
	
	\begin{lemma}[\textbf{Data-Driven  SOS Design}]\label{SOS}
		Consider the state set $\mathcal X$, the initial set $\mathcal X_a $, and the unsafe set $\mathcal X_b $, each of which is outlined by vectors of polynomial inequalities as $ \mathcal{X}=\left\{ x  \in \mathbb{R}^n \mid \mathcal{Q}(x) \geq 0\right\}$, $\mathcal X_{a}=\left\{x \in \mathbb{R}^n \mid \mathcal{Q}_{a}(x) \geq\right.$ $0\}$, and $\mathcal X_{b}=\left\{x \in \mathbb{R}^n \mid \mathcal{Q}_{b}(x) \geq 0\right\}$, respectively. Then, $\mathcal{B}(\mathbf{x})= x^\top P x + {\kappa} \sum_{i=1}^{h} \lambda^i x^\top_{i} P x_{i}$ is an RK-CBC for dt-IAUPS-td in~\eqref{eq: dt-NPS}, {and $u= \widetilde{\mathcal{F}}_1(x,x_h) P x +\widetilde{\mathcal{F}}_2(x,x_h) P x_h = \mathcal{F}_1(x,x_h) x+\mathcal{F}_2(x,x_h) x_h$ is its R-SC} if there exist a matrix $P \succ 0$, state-dependent matrices $\widetilde{\mathcal{F}}_1(x,x_h)$, $\widetilde{\mathcal{F}}_2(x,x_h)$, constants $ \mu_1, \mu_2, \eta, \beta  \in \mathbb{R}^{+}$, with $\beta>\eta$ {and $(1+\frac{1}{\mu_1}+\frac{1}{\mu_2})\|\sqrt{P}\|^2\delta\le(1-\lambda)\beta$}, $\Delta = (1+\mu_1)(1+\mu_2)-1$, $\kappa, \lambda \in (0,1)$, an SOS polynomial $\alpha: \mathcal{X}^2\!\to\!\mathbb{R}^+$, and vectors of SOS polynomials {$\Phi(x,x_h)$,$\Phi_h(x,x_h)$,} $\Phi_{a}(x)$, $\Phi_{b}(x) $, such that
		\begin{subequations}\label{LEMMA1}
			\begin{align}\label{L1} 
				-&x^\top P x - \Phi^\top_a(x)\mathcal{Q}_a(x)+ {\dfrac{\eta}{1+\kappa(\dfrac{\lambda-\lambda^{h+1}}{1-\lambda})}}, \\\label{L2} 
				&x^\top P x - \Phi^\top_b(x)\mathcal{Q}_b(x) - \beta, \\\label{L3} 
				& \widetilde{\mathcal{H}} {+} \alpha(x,x_h)\,\widetilde{\mathcal{S}}\!-\! \Big(\! \Phi^\top(x,x_h )\mathcal{Q}(x) \!+\! \Phi^\top_h(x,x_{h} )\mathcal{Q}(x_h)\!\Big)\mathds{I}_{\varepsilon \times\varepsilon},
			\end{align}
			are SOS polynomials with $\varepsilon={4} n+4(N+M)$, and $\widetilde{\mathcal{H}}$ and $\widetilde{\mathcal{S}}$ as in \eqref{T2}.
		\end{subequations}
	\end{lemma}
	\begin{proof}
		Since $\Phi_{a}(x)$ is a vector of SOS polynomials, we have
		$\Phi_{a}^\top(x)\mathcal Q_{a}(x)\ge 0$ on
		$\mathcal X_a$. Given that  \eqref{L1} is also  an SOS polynomial, hence nonnegative on $\mathcal X_a$, one has
		$$-x^\top P x + {\dfrac{\eta}{1+\kappa(\dfrac{\lambda-\lambda^{h+1}}{1-\lambda})}}
		~\ge~
		\Phi_{a}^\top(x)\,\mathcal Q_{a}(x)
		~\ge~0,$$
		which results in $\mathcal B(\mathbf x)\le \eta$ on $\mathcal X^{h+1}_a$. Similarly, one can show that the feasibility of \eqref{L2}
		implies \eqref{Th:con2}.
		As the last step of the proof, since 
		$\Phi(x,x_h)$ and $\Phi_h(x,x_h)$ are both vectors of SOS polynomials, and given that \eqref{L3} is an SOS matrix, \emph{i.e.,}
		{\begin{equation*}
				\widetilde{\mathcal{H}}  {+} \alpha(x,x_h)\,\widetilde{\mathcal{S}}-\! \Big(\! \Phi^\top\!(x,x_h )\mathcal{Q}(x) + \Phi^\top_h(x,x_{h} )\mathcal{Q}(x_h)\!\Big)\mathds{I}_{\varepsilon \times\varepsilon} \succeq 0,
		\end{equation*}}
		then \eqref{Th:con3} holds, which concludes the proof.
	\end{proof}
	
	{\begin{remark}\label{rem:cons}
			Besides the sources of conservatism from standard mathematical tools, namely Young's inequalities, the S-procedure, and the SOS multiplier degrees in~\eqref{LEMMA1}, the only framework-specific source is the structure of~\eqref{RK-CBC} (shared $P$, geometric weights), which may exclude more general RK-CBCs (potentially leading to tighter level sets or infeasibility), yet deliberately achieves the delay-independent program size of Remark~\ref{Complex}. Overall, since conditions~\eqref{Th: cons} and~\eqref{LEMMA1} are just sufficient for~\eqref{eq: CBC}, infeasibility does not rule out the existence of an RK-CBC.
	\end{remark}}
	
	\begin{algorithm}[t!]
		\caption{Data-driven design of RK-CBC and its R-SC}\label{alg}
		\begin{algorithmic}[1]
			\REQUIRE Safety specification $\Theta \!=\! (\mathcal{X}_a, \mathcal{X}_b)$, $\delta$ in \eqref{eq:Wdelta}, upper bounds on the maximum degrees of $\mathcal{M}(x)$ and $\mathcal G(x,x_h)$
			\STATE {Collect ${\mathit{U}_{-}}, {\mathit{X}_{-}} ,{\mathit{X}_{h}}, {\mathit{X}_{+}} $} as in~\eqref{eq: datarep1}
			\STATE {Form ${\mathit{M}_{-}}, {\mathit{M}}_h$, $\mathit{G}$ in \eqref{eq: datarep2}, and $ \mathds{L}(x)$} in~\eqref{transform}
			\STATE Initialize $\mu_1,\mu_2 \in \mathbb{R}^{+}$, $\kappa, \lambda \in (0,1)$\label{init-lambda}
			\STATE Solve \eqref{L3} using  \textsf{SOSTOOLS}~\cite{prajna2004sostools} and solver \textsf{Mosek}~\cite{mosek} for $P^{-1} = \Omega$, {$\widetilde{\mathcal{F}}_1(x,x_h)$, and $\widetilde{\mathcal{F}}_2(x,x_h)$}\\
			\STATE Construct $\mathcal{B}({\mathbf{x}})=x^\top P x + \kappa \sum_{i=1}^{h} \lambda^i x^\top_{i} P x_{i}$ using $P$, and compute $\gamma= (1+\frac{1}{\mu_1} +\frac{1}{\mu_2})\|\sqrt{P}\|^2$
			\STATE Design level sets $\eta, \beta$ in conditions \eqref{L1} and \eqref{L2} using the constructed $\mathcal{B}(\mathbf{x})$\\
			\ENSURE  Guaranteed robust safety {($\Upsilon\models\Theta$)} when $(1 + \frac{1}{\mu_1} + \frac{1}{\mu_2})\Vert\sqrt{P}\Vert^2\delta \leq  \beta (1-\lambda)$, with RK-CBC $\mathcal{B}(\mathbf{x})=x^\top P x + \kappa \sum_{i=1}^{h} \lambda^i x^\top_{i} P x_{i}$, and R-SC {$u= \widetilde{\mathcal{F}}_1(x,x_h) P x+ \widetilde{\mathcal{F}}_2(x,x_h) P x_h$}
		\end{algorithmic}
	\end{algorithm}
	
	\section{Case Study}\label{sec: Case}
	We validate our data-driven framework through a set of case studies, which include an academic system, a jet engine compressor~\cite{Tabuada-Jet}, and a spacecraft~\cite{khalil2002control}, while adapting them to accommodate delays and disturbances. The main goal is to synthesize an RK-CBC and its corresponding R-SC by following the step-by-step procedure outlined in Algorithm~\ref{alg}. A summary of the key results is presented in Table~\ref{tab:system-configurations}. All simulations were conducted on a \textsf{MacBook} {Pro with an \textsf{M2 Max} chip} and 32~\textsf{GB} of memory.
	
	\begin{table*}[t!]
		\centering
		\caption{{Overview of the results for unknown dt-IAUPS-td systems: $\mathtt{T}$ denotes the number of collected data points, $n$ the state dimension, $\delta$ the disturbance bound in~\eqref{eq:Wdelta}, $h$ the time delay, $\gamma$ and $\lambda$ as defined in~\eqref{subeq: decreasing}, $\eta$ and $\beta$ specify the initial and unsafe bounds, $\kappa$, $\mu_1$, and $\mu_2$ are tuning parameters, {$\varepsilon = 4n + 4(N+M)$ is the dimension of the matrix condition in~\eqref{L3}, ``SDP constr.'' and ``SDP vars.'' the number of constraints and of scalar matrix-variable entries of the resulting semidefinite program after the SOS parsing, and ``Time'' the total synthesis time, comprising the SOS parsing in \textsf{SOSTOOLS} and the SDP solution by \textsf{Mosek}. Note that none of the program dimensions depends on the delay $h$.}}}
		\label{tab:system-configurations}
		\resizebox{\linewidth}{!}
		{\begin{tabular}{@{}llcccccccccccccccc@{}}
				\toprule
				Case study & $\mathtt{T}$ & $\delta$& System degree&  $n$& $h$ & $\eta$ & $\beta$ & $\gamma$ & $\mu_1$ & $\mu_2$ & \(\kappa\) & $\lambda$ & {$\varepsilon$} & {SDP constr.} & {SDP vars.} & {Time [s]} \\ 
				\midrule\midrule
				Academic system
				& $10$ & $ 1.8 \times 10^{-3}$ & $2$ &  $2$ & $3$ &  $36.41$ &   $    40.43$ &  {$20.44$}  &  $0.59$ & $0.92$ & $0.38$ & $0.94$ & {$36$} & {$16386$} & {$32914$} & {$95.1$}\\
				\midrule
				Jet engine compressor 
				& {$12$} & $ 8 \times 10^{-4}$  & $3$ & $2$ &   $4$  &  {$51.6$} &  {$ 52.95$} &  {$7.42$} &  $0.63$  & $0.92$ & $0.23$  & $0.91$ & {$36$} & {$19026$} & {$41634$} & {$125.2$} \\
				\midrule
				Spacecraft      
				& $13$ & $1.2 \times 10^{-3}$  & $2$ & $3$ &   $3$  & $ 1.07 \times 10^{3} $  & $ 1.1 \times 10^{3}$  & {$171.15$} & $0.68$ & $0.96$ & $0.41$  & $0.93$   & {$52$} & {$28620$} & {$47612$} & {$135.6$} \\
				\bottomrule
		\end{tabular}}
	\end{table*}
	\subsection{Case Study 1: Academic System}\label{Case_study_1}
    	\begin{figure*}[t]
		\centering
		\subfloat[\label{fig:Aa}]{
			\includegraphics[width=0.3\textwidth]{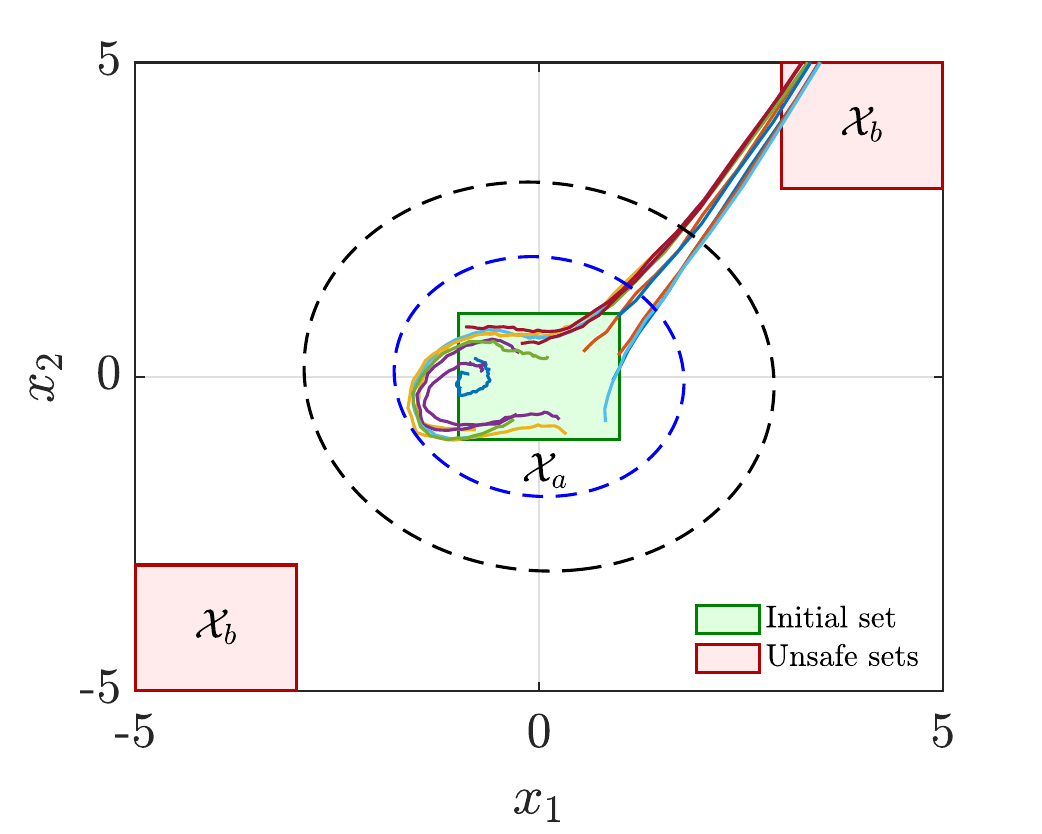}
		}\hspace{0.1cm}
		\subfloat[\label{fig:Ab}]{
			\includegraphics[width=0.3\textwidth]{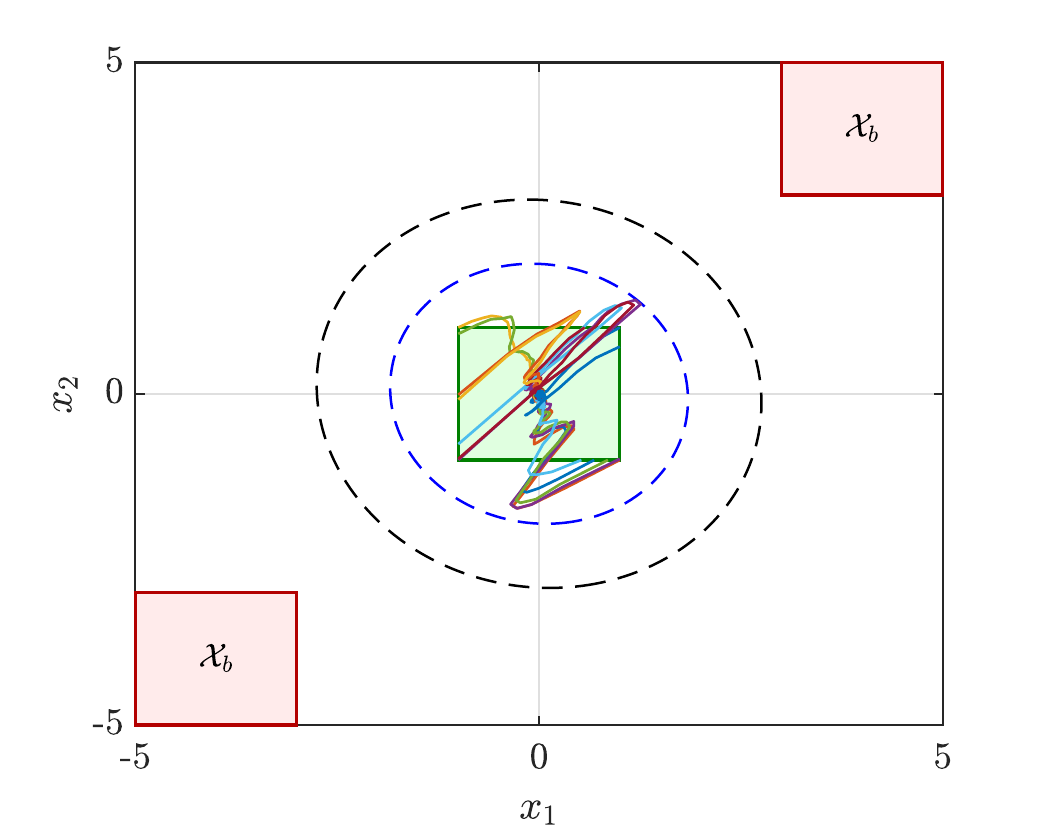}
		}\hspace{0.1cm}
		\subfloat[\label{fig:Ac}]{
			\includegraphics[width=0.3\textwidth]{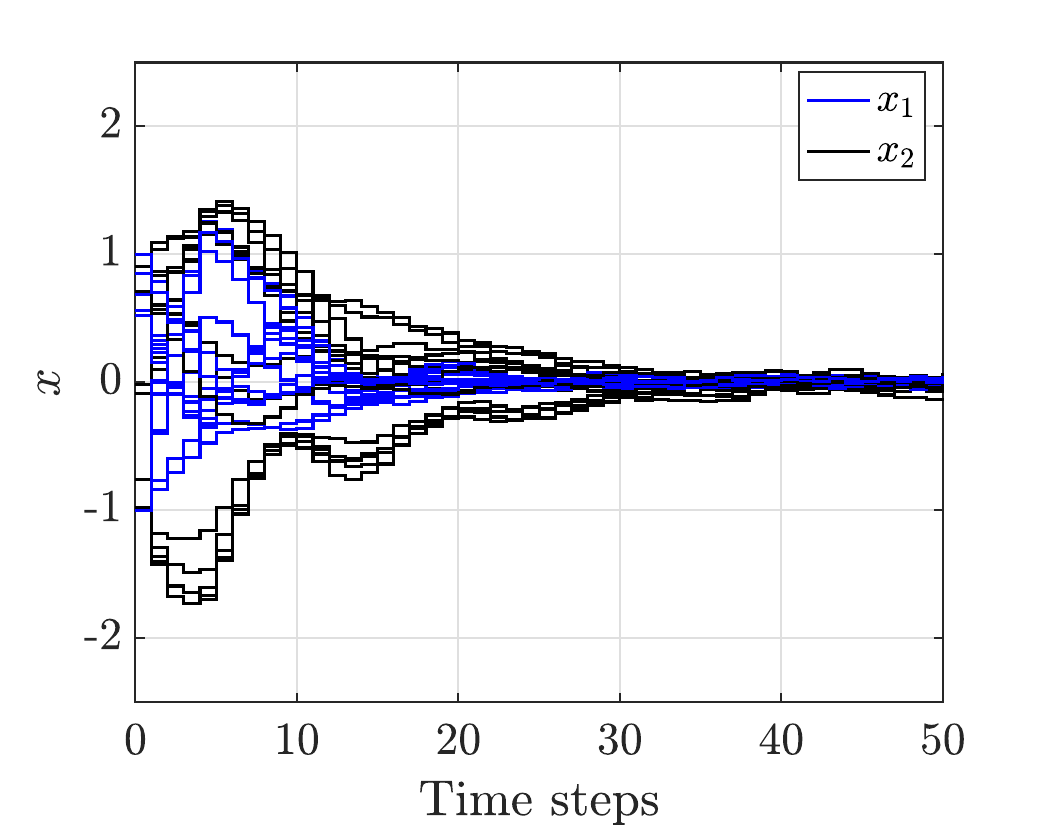}
		}
		\caption{{\textbf{Academic system}. Plot (a) illustrates the open-loop trajectories, while plot (b) displays the trajectories under the designed R-SC in~\eqref{Cont}, {initialized on the boundary of $\mathcal{X}_a = [-1,1]^2$, including its corners}.
				The blue dashed curve depicts the initial bound $x^{\top}Px = \frac{\eta}{1+\kappa(\frac{\lambda-\lambda^{h+1}}{1-\lambda})}$ in~\eqref{Th:con1}, and the black dashed curve depicts the unsafe bound $x^{\top}Px = \beta$ in~\eqref{Th:con2}. The simulations are generated using $15$ different arbitrary disturbance trajectories satisfying~\eqref{eq:Wdelta}, indicating the robustness of the proposed framework to disturbances. Plot (c) depicts the trajectories over $50$ time steps, demonstrating compliance with the specified safety property~$\Theta$.}
			\label{fig:Atraj}}
	\end{figure*}
	We consider a system with the dynamics described as
	\begin{align}\notag
		{x}_1(k+1)&\!=\! x_1(k) \!+\! 0.1x_2(k)\!+\!0.05 x_1(k\!-\!3) \!+\! 0.06x_1^2(k\!-\!3)\\\notag &\,\,\, \,\, + 0.1u(k)+ w_1(k),\\\notag
		{x}_2(k+1)& \!=\!x_2(k) +0.1x_1(k) \!+\! 0.1x_1^2(k)+0.04x_1^2(k\!-\!3) \\\label{academic}&\,\,\, \,\,  \!+\!0.02 x_1(k\!-\!3)\!+\!0.01x_2(k\!-\!3)\!+\! 0.1u(k)\!+\! w_2(k),
	\end{align}
	where $h=3$. The system in \eqref{academic} can be expressed in the form of dt-IAUPS-td in Definition~\ref{def: dt-NPS} along with its relevant matrices as
	{\begin{align*}
			{A}_1&= \begin{bmatrix}
				1  &  0.1 & 0 \\
				0.1 & 1   & 0.1
			\end{bmatrix}\!\!,~~
			A_2=	\begin{bmatrix}
				0.05 & 0 &0.06 \\
				0.02 & 0.01 & 0.04
			\end{bmatrix}\!\!,\\ 
			B &= 	\begin{bmatrix}
				0.1 \\
				0.1
			\end{bmatrix}\!\!,~~
			\mathcal{M}(x)= [x_1; x_2; x_1^2 ],~~ \mathcal{G}=1.
	\end{align*}}
	The regions of interest are given by {$\mathcal X= [-5,5]^2$, $\mathcal X_{a} = [-1,1]^2$, $\mathcal X_{b} = [3,5]^2 \cup [-5,-3]^2.$}
	We assume that the {matrices $A_1$, $A_2$, and $B$,} as well as the exact form of $\mathcal{M}(x)$, are all unknown. The only available prior information is that $\mathcal{M}(x)$ is a polynomial of degree at most $2$ and {$\mathcal{G}(x,x_h)=[1; x_1x_{1_h}]$ (\emph{i.e.,} $x_1x_{1_h} =x_1(k)x_1(k-h)$)}. The disturbance $w$ is assumed to be norm-bounded, with an upper bound given by $\delta = 1.8 \times 10^{-3}$. 
	
	We construct a polynomial dictionary for \(\mathcal{M}(x)\) that spans all monomials up to degree \(2\), namely
	$$\mathcal{M}(x)= [x_1; x_2; x_1 x_2; x_1^2 ; x_2^2].$$
	Through the satisfaction of condition \eqref{L3}, following the steps outlined in Algorithm~\ref{alg}, we compute the RK-CBC matrix
	{\begin{align*}
			P =	\begin{bmatrix}
				5.33 & 0.23\\ \star & 4.71
			\end{bmatrix}\!\!,
	\end{align*}}
	and the corresponding R-SC as
	\begin{align}\notag
		u(k)&= 0.23x_1^{2}(k) + 1.8x_1(k)x_2(k) - 2.69x_1(k) - 3.49x_2(k)\\\notag &~~~
		- 0.06x_1^{2}(k-3) + 1.41x_1(k-3)x_2(k-3)\\\label{Cont} &~~~ - 2.99x_1(k-3) - 0.31x_2(k-3).
	\end{align}
	Given the matrix \(P\) and the conditions in~\eqref{L1} and~\eqref{L2}, we design
	\(\eta =  36.41\) and \(\beta = 40.43\),
	while fulfilling the constraint \( \gamma \delta \leq \beta(1-\lambda) \) in~\eqref{eq: climit}, with
	\(\mu_1 = 0.59\), \(\lambda = 0.94\), \(\mu_2 = 0.92\), {\(\gamma = 20.44\)}, and $\kappa =0.38$, which results in
	\begin{equation*}
		\dfrac{\eta}{1+\kappa(\dfrac{\lambda-\lambda^{h+1}}{1-\lambda})}=  18.12.
	\end{equation*}
	Then, under Theorem~\ref{thm: model-based}, the system in \eqref{academic} is guaranteed to be \emph{robustly} safe in the presence of both an unknown-but-bounded disturbance and a time-invariant delay. The simulation results for the system in \eqref{academic} are presented in Fig.~\ref {fig:Atraj}.
	
	{The initial states in Fig.~\ref{fig:Ab} lie on the boundary of $\mathcal{X}_a$ (so $\mathbf{x}_0 \in \mathcal{X}_a^{h+1}$), and are thus covered by Theorem~\ref{thm: model-based}. The trajectories then decay toward a neighborhood of the origin because~\eqref{subeq: decreasing} together with~\eqref{eq:Wdelta} gives $\limsup_{k\to\infty}\mathcal{B}(\mathbf{x}(k))\le\frac{\gamma\delta}{(1-\lambda)}$. {The resulting distance between the trajectories and $\mathcal{X}_b$} is not imposed by the synthesis: $\mathcal{X}_a$ and $\mathcal{X}_b$ are prescribed by $\Theta$, and trajectories may operate nearer $\mathcal{X}_b$ by enlarging $\mathcal{X}_a$ in the specification, by increasing $\beta$ relative to $\eta$, {or by choosing $\lambda$ closer to $1$, which weakens the decay at the price of a more restrictive~\eqref{eq: climit} and~\eqref{Th:con1}}.}
	
	\subsection{Case Study 2: Jet Engine Compressor}
    
	In the second case study, we consider a physical system, namely a jet engine compressor~\cite{Tabuada-Jet}, whose dynamics are
	{\begin{align}\notag
			{x}_1(k+1)& = x_1(k) - 0.1x_2(k-4) - 0.15 x^2_1(k-4)\\\notag &~~~- 0.05 x^3_1(k) +w_1(k),\\\label{jet}
			{x}_2(k+1) & =\! x_2(k) + 0.1x_1(k-4)  + 0.1u(k)  + w_2(k),
	\end{align}}
	where $x_1(k)$ denotes the mass flow, $x_2(k)$ the pressure rise, {$u(k)$ represents the control input}, and the system is subject to a time-invariant delay of $h=4$. The system described in \eqref{jet} can be represented in the dt-IAUPS-td form as outlined in Definition~\ref{def: dt-NPS}, along with its associated matrices
	{\begin{align*}
			A_1&= \begin{bmatrix}
				1 & 0 &0 &-0.05 \\
				0& 1 & 0 & 0
			\end{bmatrix}\!\!,~~
			A_2=	\begin{bmatrix}
				0 & -0.1 &-0.15 &0 \\
				0.1 & 0 & 0 &0
			\end{bmatrix}\!\!,\\ 
			B&=	\begin{bmatrix}
				0 \\
				0.1
			\end{bmatrix}\!,~~
			\mathcal{M}(x)= [x_1; x_2; x_1^2 ; x_1^3],~~ \mathcal{G}=1.
	\end{align*}}
	In contrast to the first case study, in which the maximum degree of the system monomials was $2$, the jet engine compressor exhibits a maximum monomial degree of $3$. The regions of interest are given by $\mathcal X= [-10,10]^2$, $\mathcal X_{a} = [-2,2]^2$, $\mathcal X_{b} = [5,10]^2 \cup [-10,-5]^2$. We assume that the {matrices \(A_1\), \(A_2\) and \(B\),} as well as the actual dictionary \(\mathcal{M}(x)\), are all unknown, and the only available prior information is that $\mathcal{M}(x)= [x_1; x_2; x_1 x_2; x_1^2 ; x_2^2 ; x_1^3]$, \(\mathcal{G}=1\), and \( \delta = 8\times10^{-4} \).
	
	Following the steps outlined in Algorithm~\ref{alg}, we compute the RK-CBC matrix
	{\begin{align*}
			P =	\begin{bmatrix}
				1.35 & -0.23\\ \star & 1.94
			\end{bmatrix}\!\!,
	\end{align*}}
	and the corresponding R-SC as
	{\begin{align}\label{Cont1}
			u(k)=1.6x_1(k) - 2x_2(k) - 0.8x_1^2(k) - x_1(k-h).
	\end{align}}
	{Given the matrix \(P\) and the conditions in~\eqref{L1} and~\eqref{L2}, we design
		{$\eta = 51.6$} and {$\beta=52.95$},
		while fulfilling the constraint \( \gamma \delta \leq \beta(1-\lambda) \) in~\eqref{eq: climit}, with
		\(\mu_1 = 0.63\), \(\lambda = 0.91\), \(\mu_2 = 0.92\), {\(\gamma = 7.42\)} and $\kappa =0.23$, leading to
		\begin{equation*}
			\dfrac{\eta}{1+\kappa(\dfrac{\lambda-\lambda^{h+1}}{1-\lambda})}=  {29.81}.
	\end{equation*}}
	Under Theorem~\ref{thm: model-based}, the jet engine compressor in~\eqref{jet} is guaranteed to be \emph{robustly} safe in the presence of both disturbances and delays, as illustrated by the simulation results in Fig.~\ref{fig:Btraj}.
	
	    	\begin{figure*}[t]
		\centering
		\subfloat[\label{fig:Ba}]{
			\includegraphics[width=0.3\textwidth]{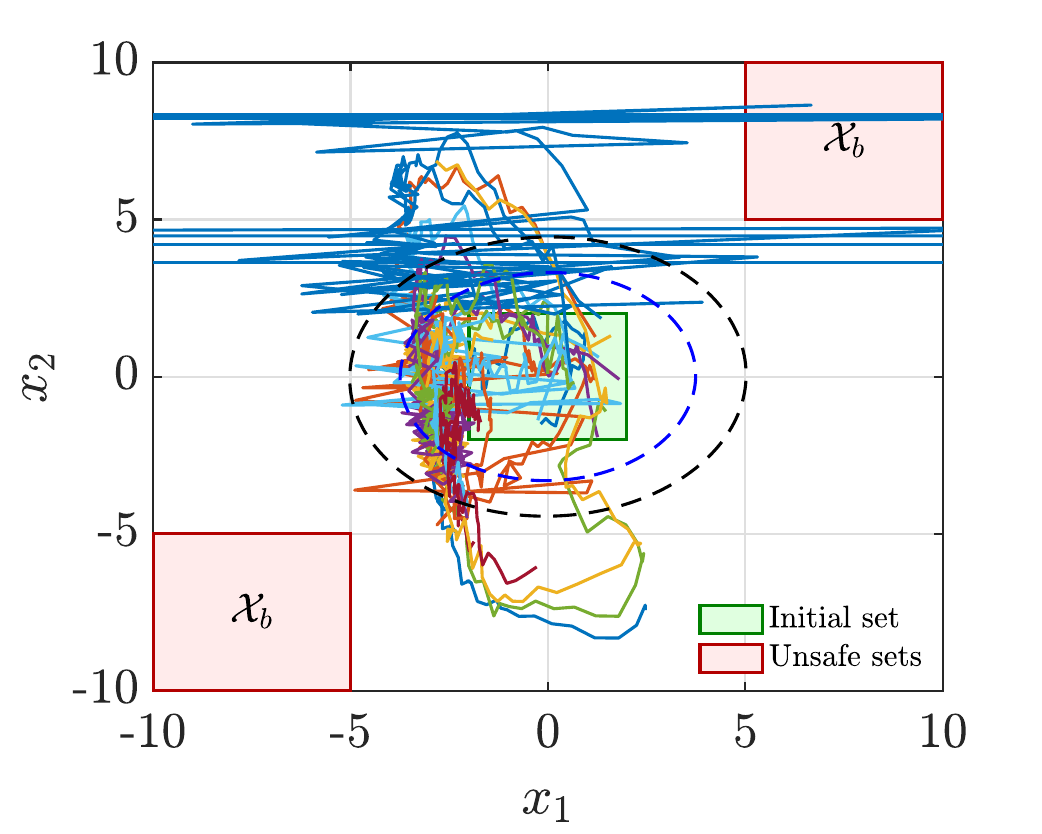}
		}\hspace{0.1cm}
		\subfloat[\label{fig:Bb}]{
			\includegraphics[width=0.3\textwidth]{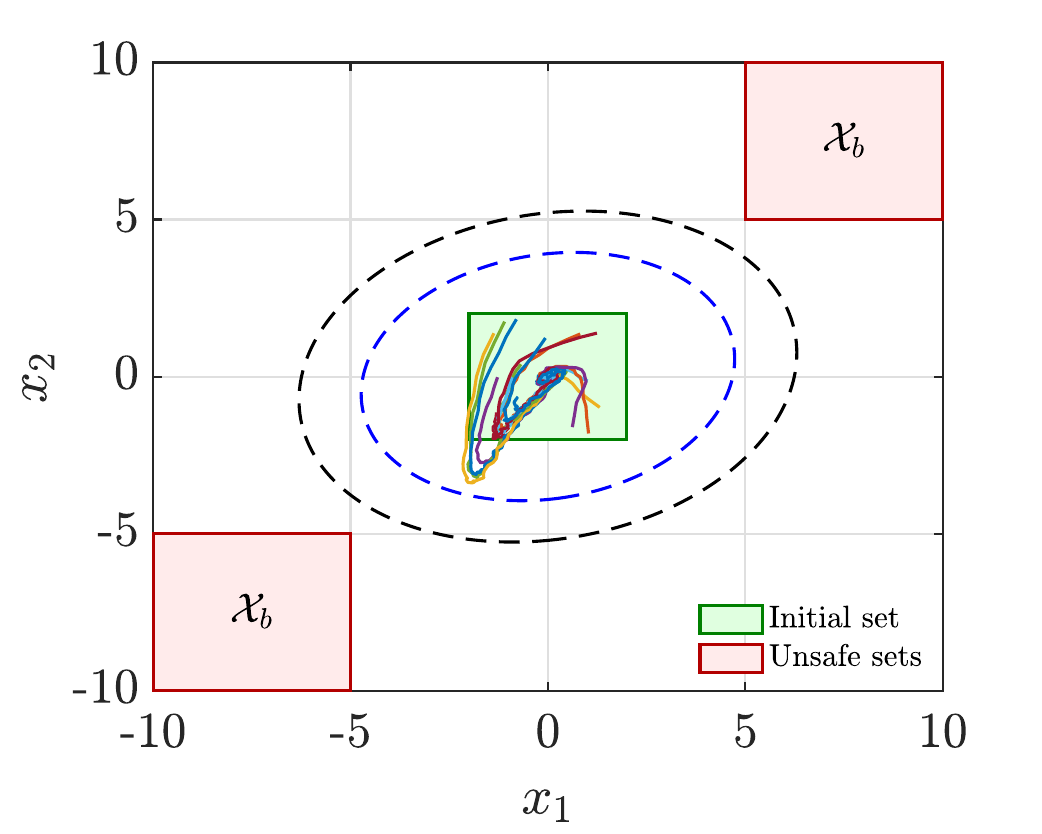}
		}\hspace{0.1cm}
		\subfloat[\label{fig:Bc}]{
			\includegraphics[width=0.3\textwidth]{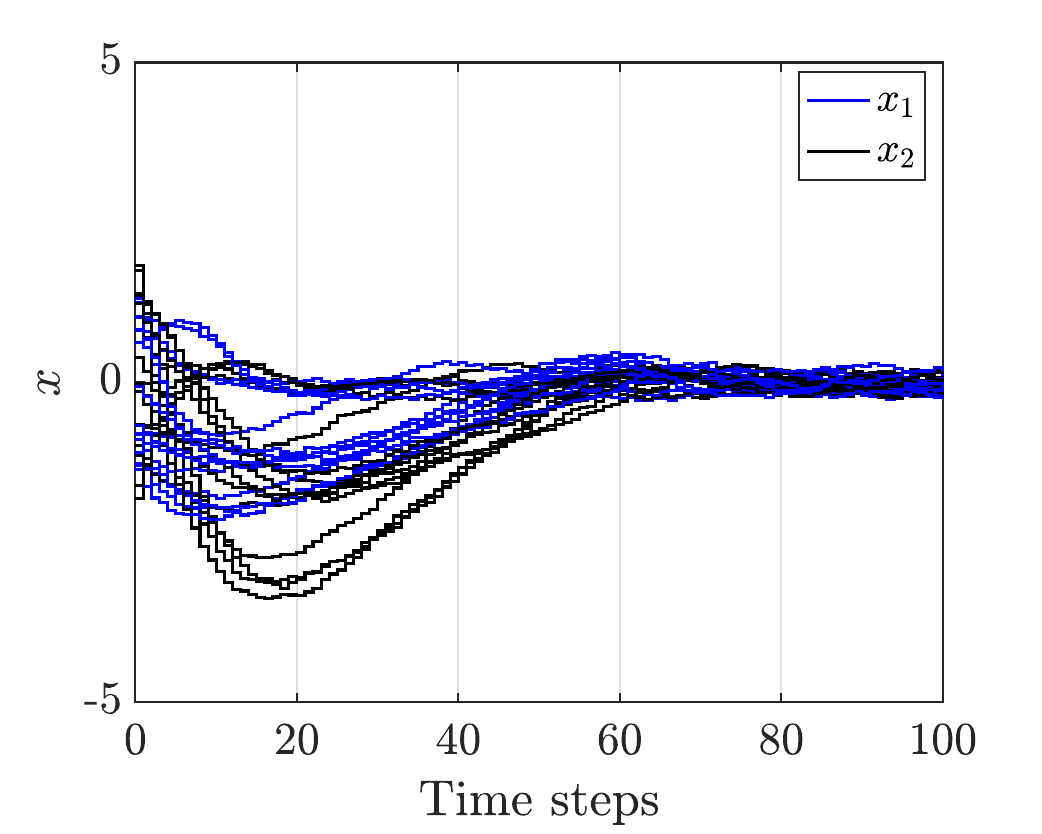}
		}
		\caption{{\textbf{Jet engine compressor}. Plot (a) shows trajectories under a random controller, while plot (b) presents trajectories with the designed safety controller in~\eqref{Cont1}, initialized from $\mathcal{X}_a = [-2,2]^2$. The blue dashed curve depicts the initial bound $x^{\top}Px = \frac{\eta}{1+\kappa(\frac{\lambda-\lambda^{h+1}}{1-\lambda})}$ in~\eqref{Th:con1}, and the black dashed curve depicts the unsafe bound $x^{\top}Px = \beta$ in~\eqref{Th:con2}. Results are obtained using $25$ disturbance realizations satisfying~\eqref{eq:Wdelta}, demonstrating robustness to disturbances. Plot (c) depicts the state evolution over $100$ time steps, confirming satisfaction of the safety specification~$\Theta$.}
			\label{fig:Btraj}}
	\end{figure*}
	
	\subsection{Case Study 3: Spacecraft}

    \begin{figure*}[t]
		\centering
		\subfloat[\label{fig:Ca}]{
			\includegraphics[width=0.3\textwidth]{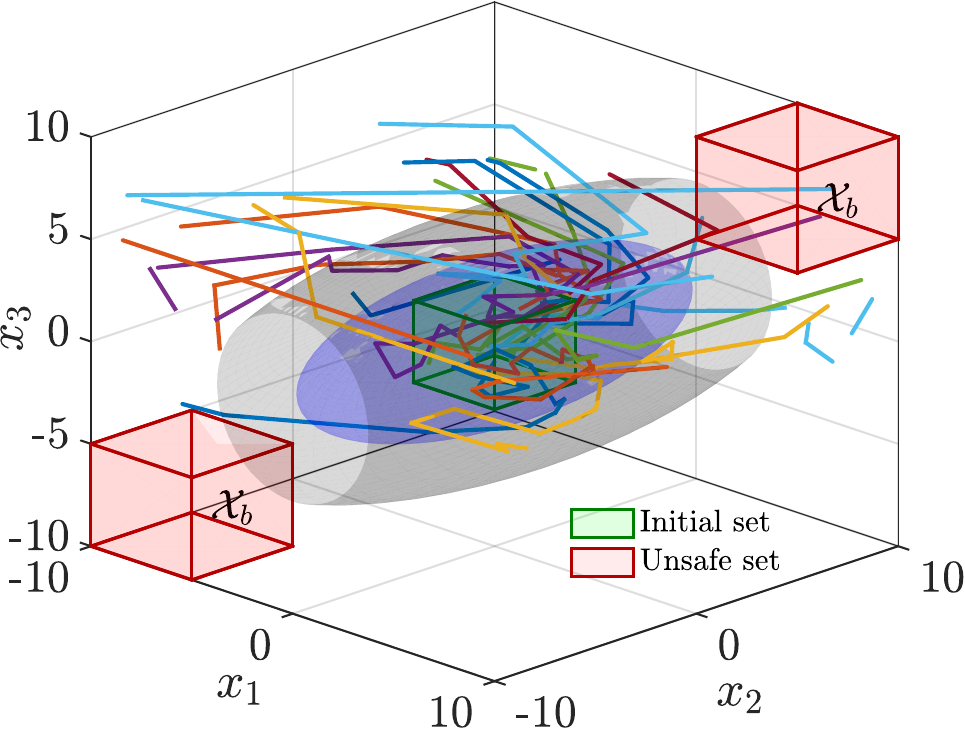}
		}\hspace{0.1cm}
		\subfloat[\label{fig:Cb}]{
			\includegraphics[width=0.3\textwidth]{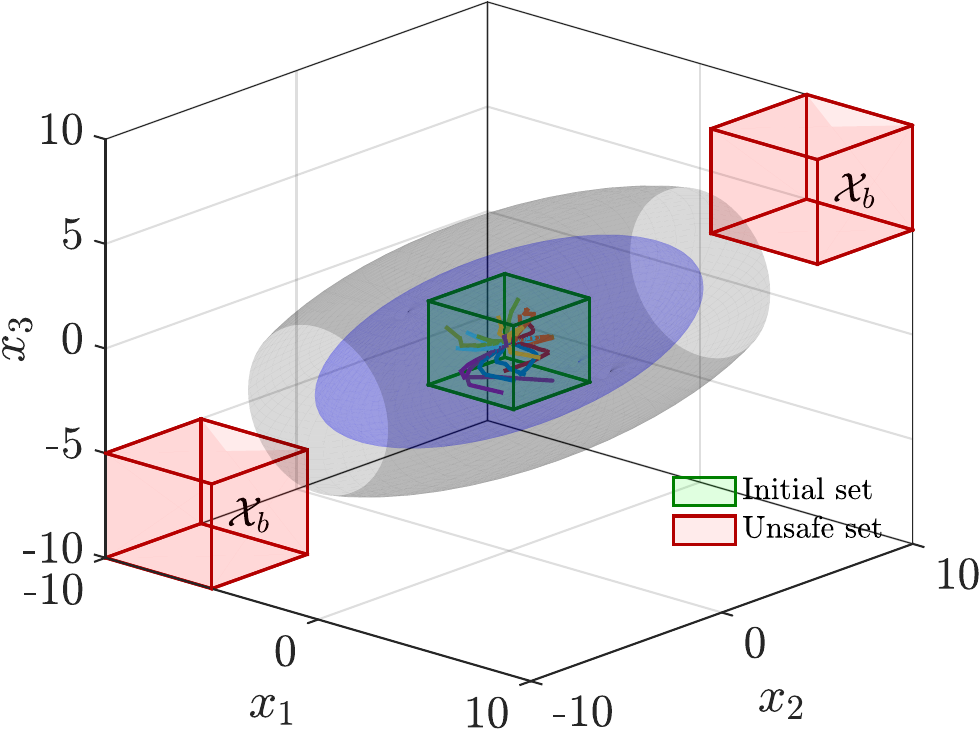}
		}\hspace{0.1cm}
		\subfloat[\label{fig:Cc}]{
			\includegraphics[width=0.3\textwidth]{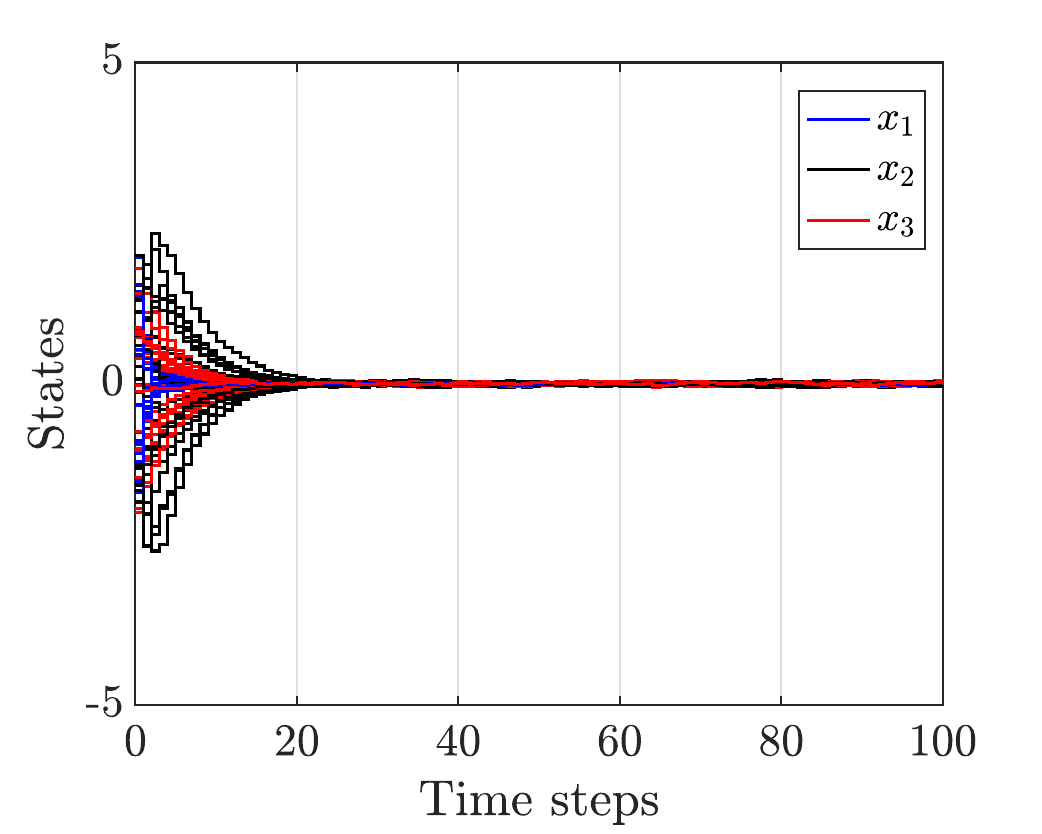}
		}
		\caption{{\textbf{Spacecraft}. Plot (a) presents trajectories under a random controller, whereas plot (b) shows trajectories with the synthesized safety controller in~\eqref{safety-c-3}, initialized from $\mathcal{X}_a = [-2,2]^3$, with the initial bound in~\eqref{Th:con1} marked by the inner blue ellipsoid and the unsafe bound in~\eqref{Th:con2} marked by the outer gray ellipsoid. The results are acquired using $25$ disturbance realizations fulfilling~\eqref{eq:Wdelta}, illustrating robustness to disturbances. Plot (c) depicts the state evolution over $100$ time steps, confirming adherence to the safety specification~$\Theta$.}
			\label{fig:Ctraj}}
	\end{figure*}
    
	As the last case study, we examine a spacecraft~\cite{khalil2002control}, whose dynamics are
	{\begin{align} \notag
			{x}_1(k+1) &= 	{x}_1(k) \!+\! \frac{\sigma_{2} - \sigma_{3}}{\sigma_{1}}x_2(k) x_3(k) 
			\!+\! \frac{1}{\sigma_1} u_1(k)  \\\notag &~~~  + w_1(k), \\\notag
			{x}_2(k+1) &=  {x}_2(k) + \frac{\sigma_{3} - \sigma_{1}}{\sigma_{2}} x_1(k-3) x_3(k-3)\\\notag &~~~ 
			+ \frac{1}{\sigma_{2}} u_2(k) + w_2(k), \\\notag
			{x}_3(k+1) &=  {x}_3(k) + \frac{\sigma_{1} - \sigma_{2}}{\sigma_{3}}\, x_1(k-3)\, x_2(k-3)\\\label{space}&~~~  
			+ \frac{1}{\sigma_{3}} u_3(k) + w_3(k),
	\end{align}}
	with time delay $h=3$. Additionally, $\sigma_{1}$ to $\sigma_{3}$
	are the principal moments of inertia, with $\sigma_1=\sigma_2=20$ and $\sigma_3=30$, {and $u(k)=\left[u_1(k) ; u_2(k) ; u_3(k)\right]$ is the torque input}. The spacecraft in \eqref{space} can be expressed in the form of dt-IAUPS-td in Definition~\ref{def: dt-NPS}, along with its relevant matrices as
	{\begin{align*}
			A_1 &\!=\!\begin{bmatrix}  1& 0  & 0&\frac{\sigma_{2} - \sigma_{3}}{\sigma_{1}} & 0 &0 \\
				0 & 1& 0  & 0& 0& 0 \\
				0& 0 & 1 & 0& 0& 0\end{bmatrix}\!\!, ~~ B \!= \! \begin{bmatrix} \frac{1}{\sigma_{1}}  & 0  & 0\\ 0 &  \frac{1}{\sigma_{2}} &0 \\ 0 &0 &\frac{1}{\sigma_{3}}    \end{bmatrix}\!\!,\\
			A_2 &\!=\! \begin{bmatrix}  0& 0 & 0 & 0& 0& 0 \\ 0& 0 & 0& 0 &  \frac{\sigma_{3} - \sigma_{1}}{\sigma_{2}}  & 0  \\0& 0&0&0&0& \frac{\sigma_{1} - \sigma_{2}}{\sigma_{3}} \end{bmatrix}\!\!,\\ ~~ \mathcal{M}(x) &\!=\! [x_1; x_2; x_3; x_2 x_3; x_1 x_3 ; x_1 x_2 ], ~~  \mathcal{G}\!=\!\mathds{I}_3.
	\end{align*}}
	\begin{figure*}[t]
	\centering
	\begin{subfigure}[t]{0.32\textwidth}
		\centering
		\includegraphics[width=\textwidth]{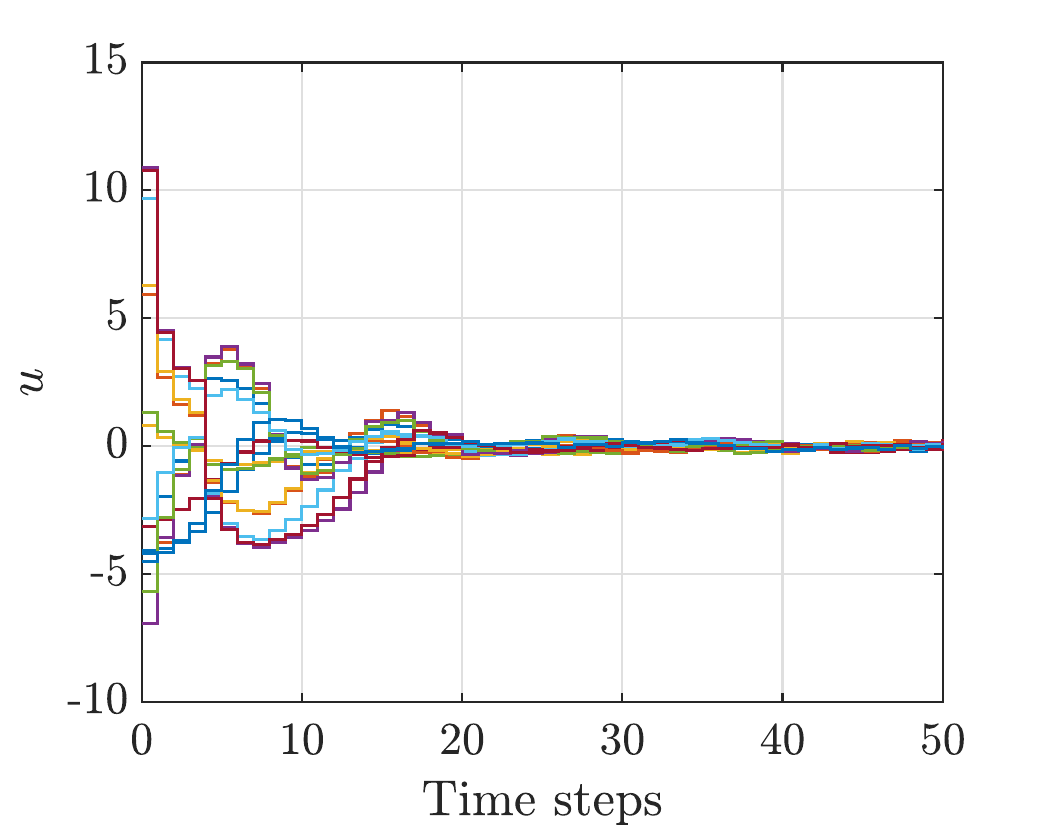}
		\subcaption{}
		\label{fig:input_academic}
	\end{subfigure}
	\hfill
	\begin{subfigure}[t]{0.32\textwidth}
		\centering
		\includegraphics[width=\textwidth]{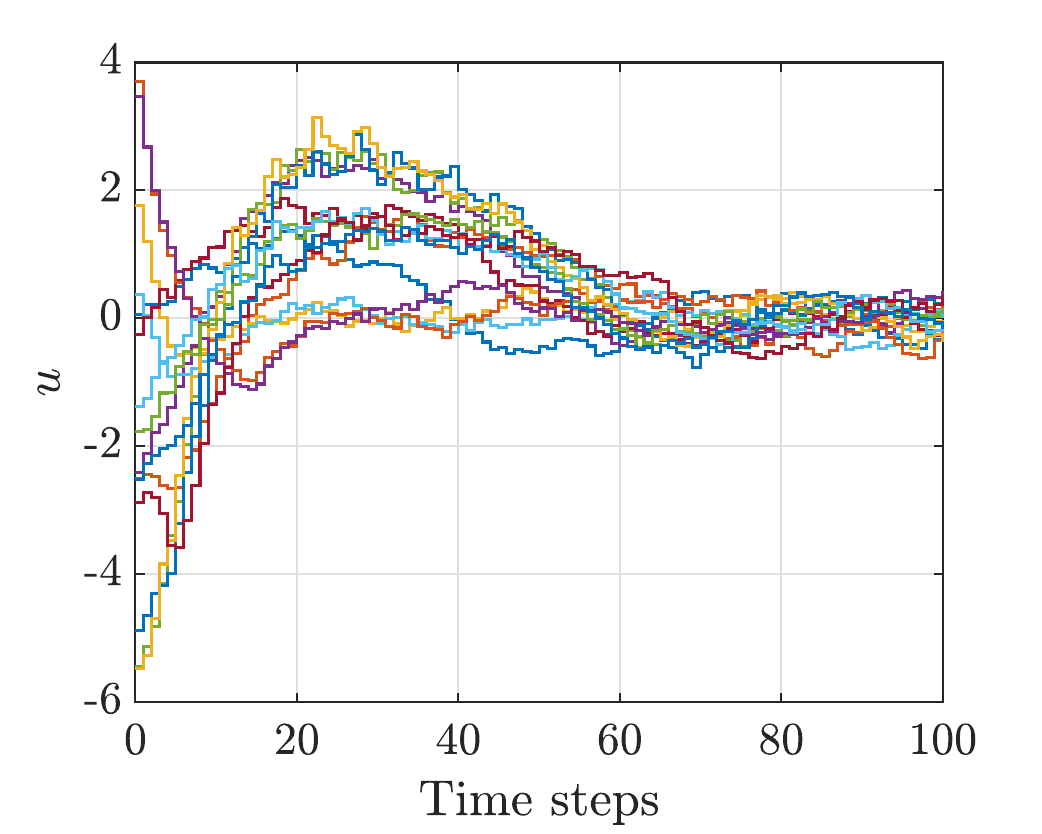}
		\subcaption{}
		\label{fig:input_jet}
	\end{subfigure}
	\hfill
	\begin{subfigure}[t]{0.32\textwidth}
		\centering
		\includegraphics[width=\textwidth]{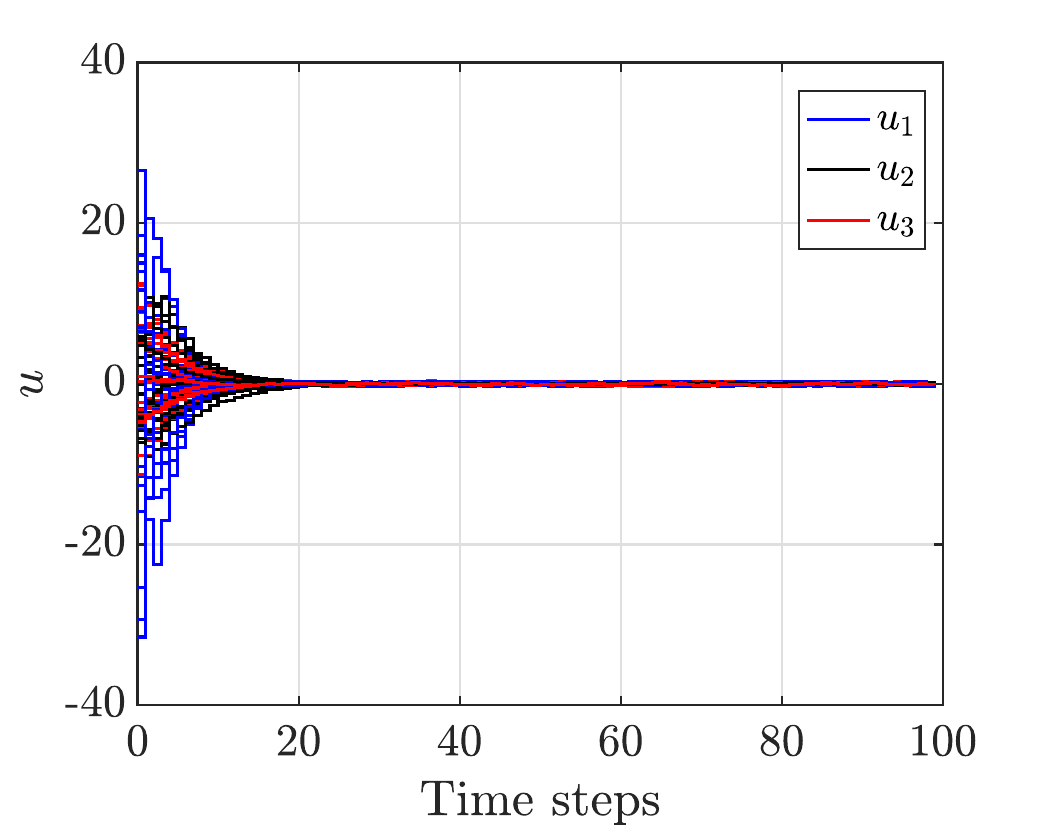}
		\subcaption{}
		\label{fig:input_spacecraft}
	\end{subfigure}
	\caption{{Control-input trajectories of the proposed R-SC for the three case
			studies: (a) academic system, input of the R-SC in \eqref{Cont};
			(b) jet engine compressor, input of the R-SC in \eqref{Cont1};
			and (c) spacecraft, inputs $u_1,u_2,u_3$ of the R-SC in
			\eqref{safety-c-3}. All trajectories are generated under the same
			disturbance realizations and initial conditions as
			Figs.~\ref{fig:Atraj}--\ref{fig:Ctraj}.}}
	\label{fig:inputs_all}
\end{figure*}

\begin{figure*}
	\centering
	
	\begin{subfigure}[t]{0.22\textwidth}
		\centering
		\includegraphics[width=\textwidth]{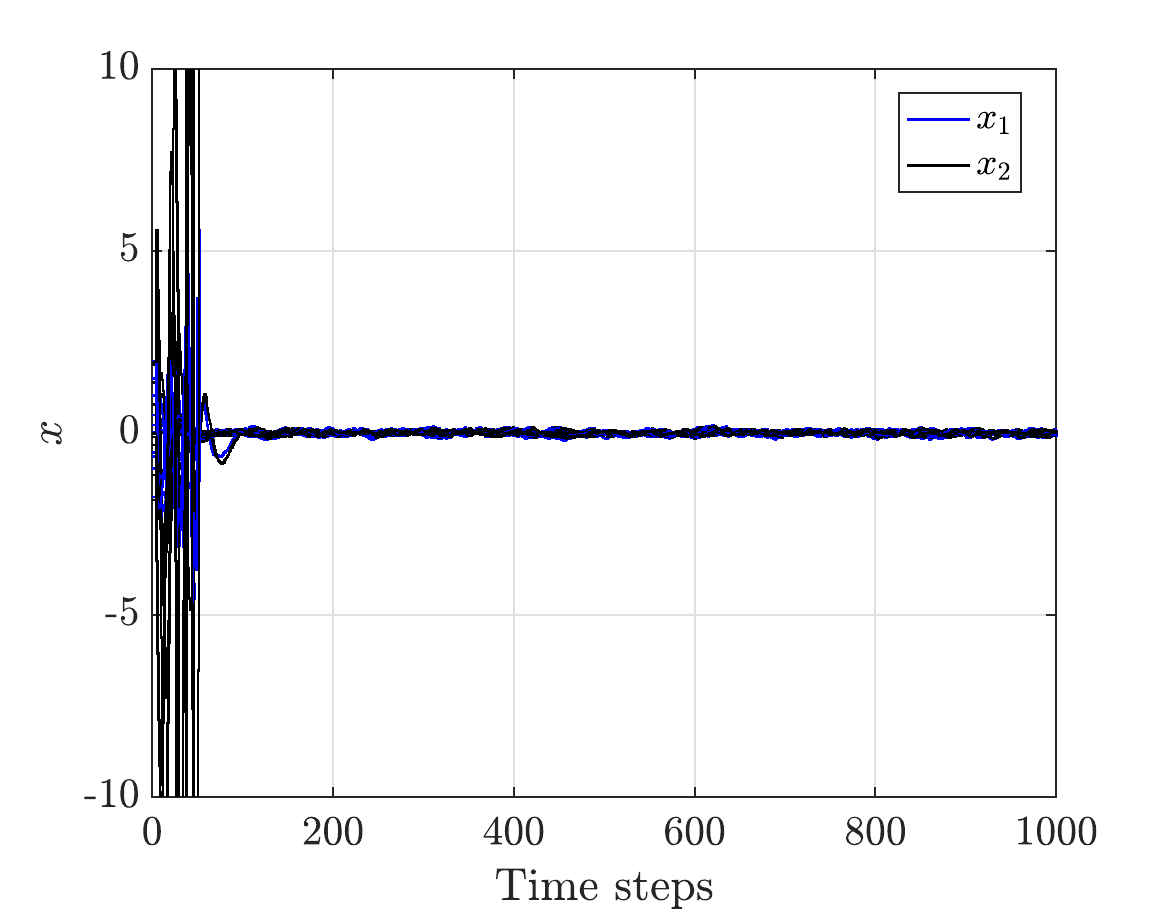}
		\caption{}
		\label{fig:jet_state_du_response}
	\end{subfigure}
	\hfill
	\begin{subfigure}[t]{0.22\textwidth}
		\centering
		\includegraphics[width=\textwidth]{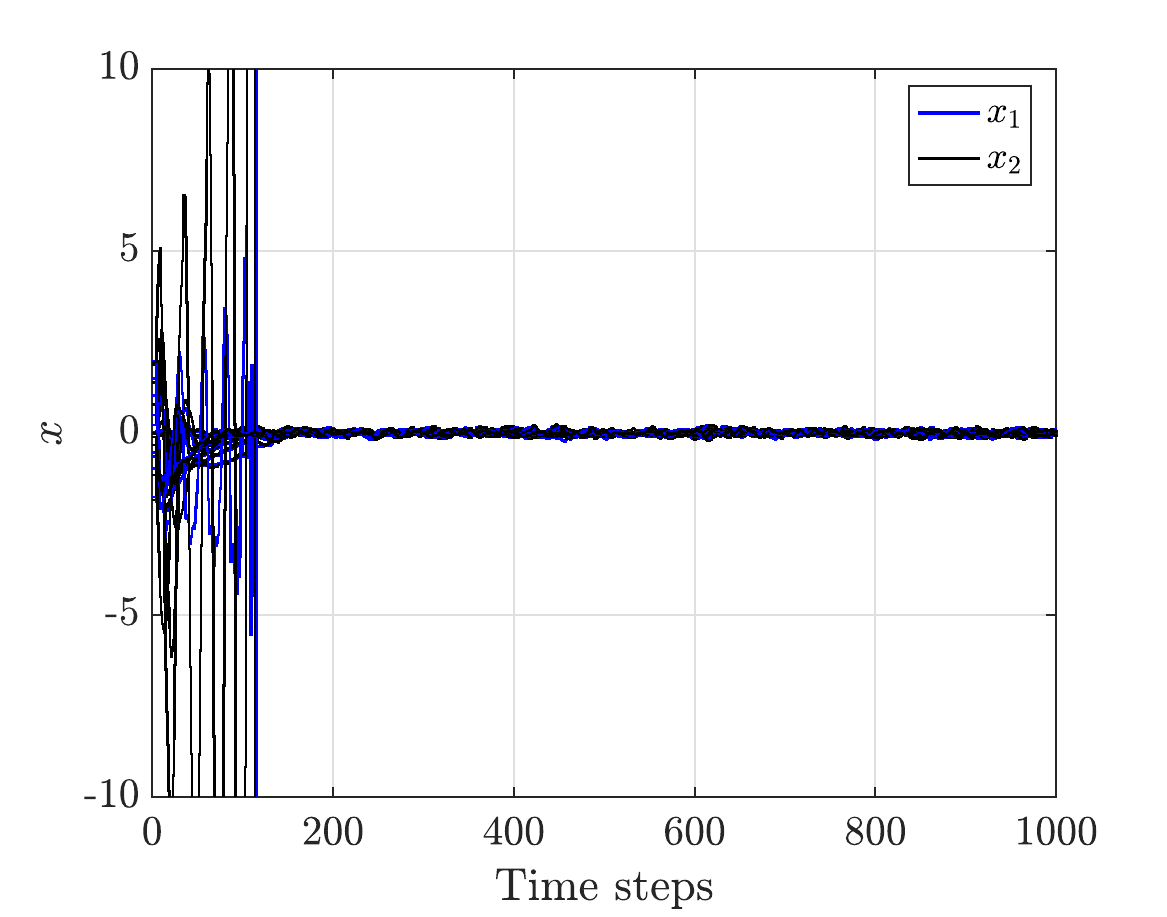}
		\caption{}
		\label{fig:jet_state_da_response}
	\end{subfigure}
	\hfill
	\begin{subfigure}[t]{0.22\textwidth}
		\centering
		\includegraphics[width=\textwidth]{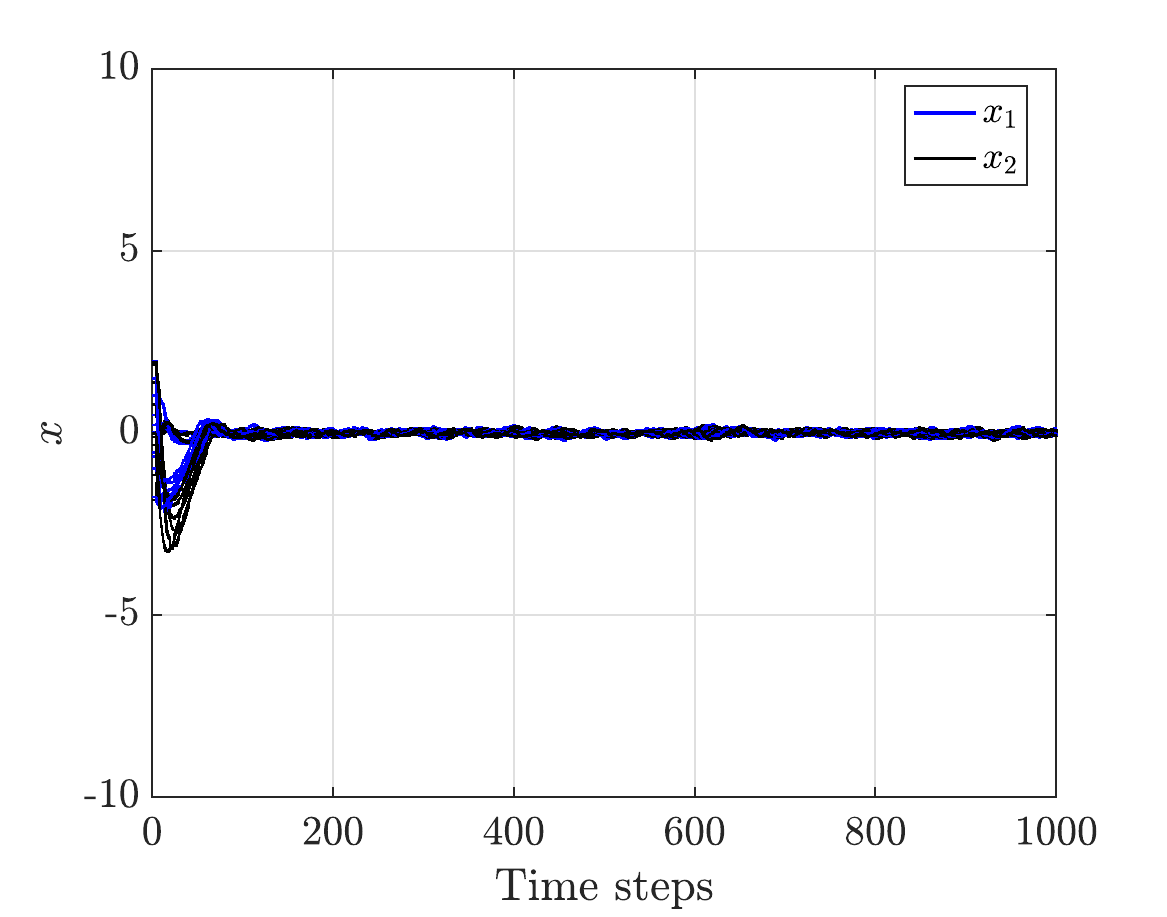}
		\caption{}
		\label{fig:jet_state_rsc_response}
	\end{subfigure}
	\hfill
	\begin{subfigure}[t]{0.22\textwidth}
		\centering
		\includegraphics[width=\textwidth]{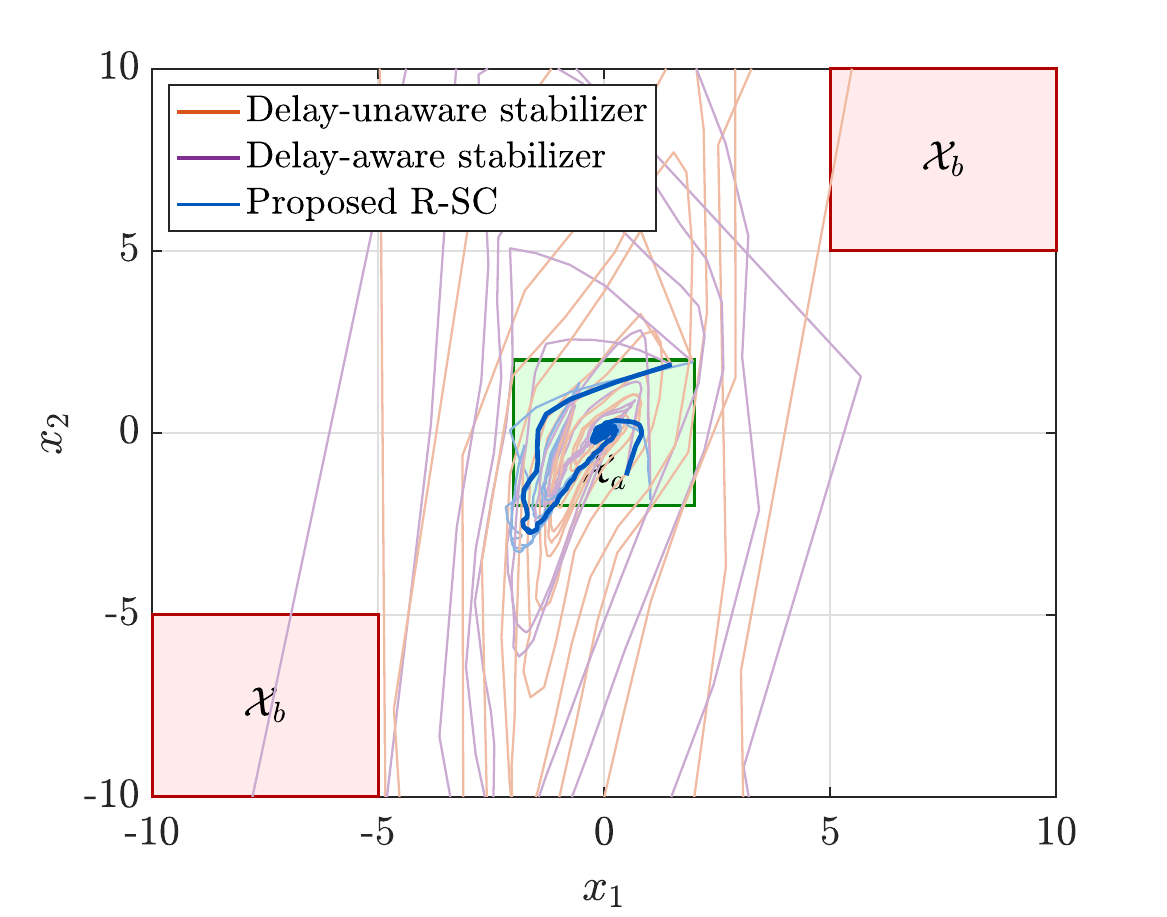}
		\caption{}
		\label{fig:jet_phase_response}
	\end{subfigure}
	
	\caption{{Jet engine compressor comparison between stabilization and safety under identical
			initial conditions and admissible disturbances:
			(a) delay-unaware stabilizer, (b) delay-aware stabilizer, and (c) proposed
			R-SC (state trajectories), with the comparison in (d). The
			stabilizing baselines may enter \(\mathcal{X}_b\), whereas the proposed
			R-SC keeps all trajectories outside it.}}
	\label{fig:jet_all_in_one_response}
\end{figure*}
	Unlike the first two case studies, which feature two-dimensional dynamics, the spacecraft’s dynamics are three-dimensional. We assume that the {matrices $A_1$, $A_2$, and $B$,} as well as the exact $\mathcal{M}(x)$, are all unknown, with the only prior information being that $\mathcal{M}(x)= [x_1; x_2;x_3; x_2 x_3; x_1 x_3 ; x_1 x_2 ; x_2^2]$, $\mathcal{G}=\mathds{I}_3$, and $\delta = 1.2 \times 10^{-3}$. The regions of interest are given by $\mathcal X= [-10,10]^3$, $\mathcal X_{a} = [-2,2]^3$, $\mathcal X_{b} = [5,10]^3 \cup [-10,-5]^3$. We compute the RK-CBC matrix using Algorithm~\ref{alg} as
	{\begin{align*}
			P &=	\begin{bmatrix}
				48.73 & {0} & -0.05\\ \star & 5.73 & 0 \\ \star & \star & 33.53
			\end{bmatrix}\!\!,
	\end{align*}}
	and the corresponding R-SC as
	\begin{align}\notag
		u_1(k)&= 0.03x_1^{2}(k) + 8.92x_2(k)x_3(k) - 12.46x_1(k)\\\notag &~~~-0.15x_1(k-3)x_2(k-3),\\\notag u_2(k) &= 0.02x_1^{2}(k) - 3.73x_2(k)-0.88x_1^{2}(k-3)\\\notag &~~~ + 0.03x_1(k-3)x_3(k-3) + 0.23x_2^{2}(k-3),\\\notag u_3(k)&=-0.05x_3^{2}(k) - 6.32x_3(k)+ 0.16x_1^{2}(k-3)\\\label{safety-c-3} & ~~~+ 0.01x_3^{2}(k-3).
	\end{align}
	{Given the matrix \(P\) and the conditions in~\eqref{L1} and~\eqref{L2}, we design
		$\eta =  1.07 \times 10^{3} $ and $\beta =  1.1 \times 10^{3}$,
		while satisfying the constraint \( \gamma \delta \leq \beta(1-\lambda) \) in~\eqref{eq: climit}, with
		\(\mu_1 = 0.68\), \(\lambda = 0.93\), \(\mu_2 = 0.96\), {\(\gamma = 171.15\)} and $\kappa =0.41$, resulting in
		\begin{equation*}
			\dfrac{\eta}{1+\kappa(\dfrac{\lambda-\lambda^{h+1}}{1-\lambda})}=   518.35.
	\end{equation*}}
	Under Theorem~\ref{thm: model-based}, the spacecraft in~\eqref{space} is guaranteed to satisfy robust safety in the presence of a time-invariant delay and an unknown-but-bounded disturbance, as illustrated by the simulation results in Fig.~\ref{fig:Ctraj}.
	
	{The control-input trajectories generated by the synthesized R-SC for all three case studies are shown in Fig.~\ref{fig:inputs_all}. In each case, the inputs are produced by the polynomial state-feedback controllers in \eqref{Cont}, \eqref{Cont1}, and \eqref{safety-c-3}, under the same initial conditions and disturbance realizations used for the state trajectories in Figs.~\ref{fig:Atraj}--\ref{fig:Ctraj}. To further highlight the distinction between stabilization and safety, we compare the proposed R-SC with two stabilizing baselines on the jet engine case study. The first baseline is delay-unaware, {obtained by applying the data-driven stabilization of~\cite{guo2022datadriven} to the delay-free counterpart of~\eqref{jet} ($h=0$)}, $u(k) = 2.29x_1(k)
		-4.11x_2(k)+5x_1^3(k),$ while the second is delay-aware, synthesized from the known model via Lyapunov--Krasovskii stabilization~\cite{Fridman2014} without a safety specification, given by $u(k)
		=
		0.96x_1(k)
		-1.59x_2(k)
		+1.03x_1(k-h)
		-2.22x_2(k-h)
		+0.44x_1(k)x_1(k-h)
		+1.64x_1^3(k-h).$
		Both are nonlinear polynomial feedback laws designed to promote stabilization, but neither explicitly accounts for the unsafe set \(\mathcal{X}_b\). By contrast, the proposed R-SC in \eqref{Cont1} is synthesized from the RK-CBC conditions and explicitly certifies avoidance of \(\mathcal{X}_b\). As shown in Fig.~\ref{fig:jet_all_in_one_response}, under identical initial conditions and admissible disturbances, the stabilizing baselines can drive the system through \(\mathcal{X}_b\) during the transient, whereas the proposed R-SC keeps every trajectory outside \(\mathcal{X}_b\) for all time. This confirms that stabilization alone does not imply safety, and that the safety guarantee provided by the RK-CBC-based synthesis cannot, in general, be recovered from a stabilization-oriented design that does not incorporate the safety specification.}

	{We also compare with the delay-free safety formulation of~\cite{akbarzadeh2024learning}, applied to the $h=0$ counterpart of~\eqref{jet} with the same sets, $\delta$, data length, and degrees. The model class has no delayed channel, so on the delayed data the term $A_2\mathcal{M}(x(k-h))$ can only be treated as disturbance. {Already over $\mathcal{X}_a=[-2,2]^2$}, $\|A_2\mathcal{M}(x_h)\|^2$ reaches $0.68$, so a sound bound is {at least} $\delta_{\mathrm{eff}}=(\sqrt{0.68}+\sqrt{\delta})^2\approx0.73$, three orders of magnitude above $\delta=8\times10^{-4}$; the condition $\gamma\delta_{\mathrm{eff}}\le(1-\lambda)\beta$ then requires $\frac{\beta}{\gamma}\ge8.1$ for $\lambda=0.91$, {and $\mathtt{T}\delta_{\mathrm{eff}}$ enlarges the data-conformity set accordingly, which makes the resulting synthesis severely conservative}. Keeping the true $\delta$ instead certifies, at best, a delay-free {model fitted to the delayed data,} with no guarantee for~\eqref{jet}. Fig.~\ref{fig:REF29} shows the delay-free controller {of~\cite{akbarzadeh2024learning}} and the R-SC~\eqref{Cont1} on the delayed dynamics: only the latter keeps all trajectories outside $\mathcal{X}_b$.}

    \begin{figure}[t]
		\centering
		\includegraphics[width=0.6\linewidth]{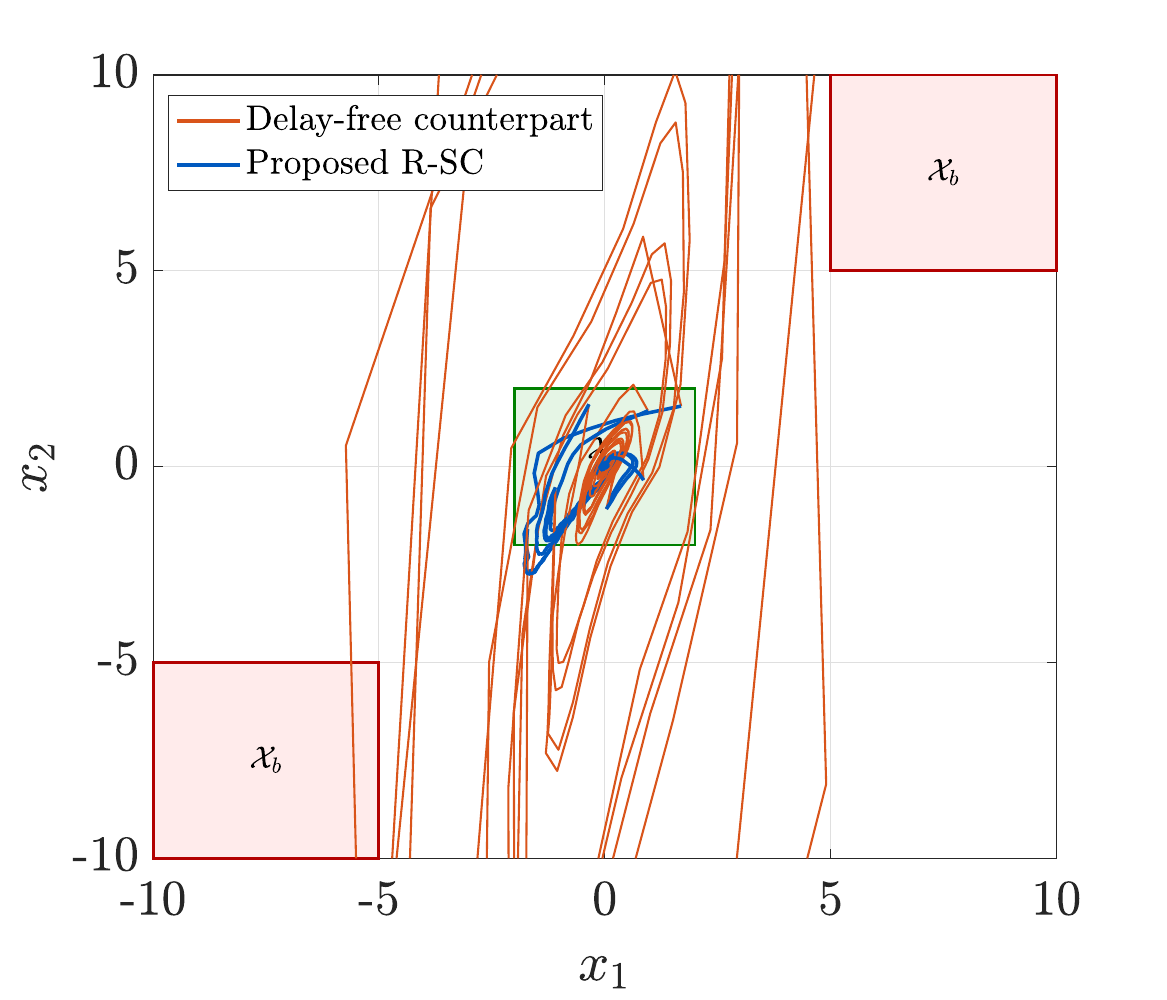}
		\caption{{\textbf{Jet engine compressor}. Closed-loop trajectories of~\eqref{jet} ($h=4$) under identical initial conditions and disturbance realizations satisfying~\eqref{eq:Wdelta}: controller synthesized by the delay-free formulation of~\cite{akbarzadeh2024learning} on the $h=0$ counterpart, and the proposed R-SC~\eqref{Cont1}. The former carries no guarantee once the delay is present; the latter keeps every trajectory outside $\mathcal{X}_b$.}}
		\label{fig:REF29}
	\end{figure}
	
	{Two features of Table~\ref{tab:system-configurations} are worth noting. First, every certificate and controller is synthesized from a single trajectory of $\mathtt{T}=10$--$13$ samples, with no i.i.d.\ requirement, and carries a deterministic guarantee. Second, $\varepsilon$ and the SDP size do not depend on the delay (Remark~\ref{Complex}): the academic system ($h=3$) and the jet engine compressor ($h=4$) share $\varepsilon=36$, {and differ in SDP size and solver time mainly because of the degree-three dictionary of the latter. The larger program size for the spacecraft stems instead from $n$, $M$, and $N$: all three enter $\varepsilon$ linearly, while $n$ also enlarges the Gram bases of the SOS multipliers, which contain $\binom{2n+d}{d}$ monomials for multipliers of degree $2d$ in the $2n$ variables $(x,x_h)$. These $2n$ variables, versus $n$ in the delay-free case, are the intrinsic cost of the delayed setting.} With degree-two dictionaries and {multipliers of degree up to four}, we expect Algorithm~\ref{alg} to remain practical up to $n\approx5$ ({Gram bases of $\binom{12}{2}=66$ monomials for $2n=10$}), while $n\gtrsim8$ is beyond the present monolithic formulation, a limitation shared with delay-free SOS synthesis.}

	\section{Conclusion}\label{sec: Conclusion}
	We developed a data-driven framework for constructing {RK-CBC} and {R-SC} for discrete-time input-affine polynomial systems with {(partially) unknown} dynamics, time-invariant delays, and unknown-but-bounded disturbances. The proposed approach addressed the dual challenges of unknown dynamics subject to unknown disturbances and delay-dependent safety requirements, by extending classical control barrier certificates to a Krasovskii-based formulation that explicitly incorporated delayed state effects, co-designed with a two-channel data-driven parameterization of the delayed dynamics. Relying solely on {a single} finite-horizon input–state {trajectory}, the framework enabled direct safety certification and controller design without requiring explicit system models, while ensuring {robust safety against bounded disturbances and time delays as long as the closed-loop trajectories remain within the state set}. The resulting synthesis problem was formulated as an SOS {feasibility} program {whose size is independent of the delay}, providing a structured and tractable design methodology. The effectiveness of the approach was demonstrated through three case studies, including a comparison with stabilizing controllers that fail to maintain safety. Developing data-driven approaches for classes of systems beyond polynomial dynamics remains a direction for future work.
	
	\bibliographystyle{ieeetr}
	\bibliography{biblio}

\begin{thebibliography}{10}

\bibitem{Srikant2004}
R.~Srikant, {\em The Mathematics of Internet Congestion Control}.
\newblock Birkh{\"a}user, 2004.

\bibitem{Seborg2004}
D.~E. Seborg, T.~F. Edgar, and D.~A. Mellichamp, {\em Process Dynamics and
  Control}.
\newblock Wiley, 2004.

\bibitem{Richard2003}
J.-P. Richard, ``Time-delay systems: an overview of some recent advances and
  open problems,'' {\em Automatica}, vol.~39, no.~10, pp.~1667--1694, 2003.

\bibitem{Fridman2014}
E.~Fridman, {\em Introduction to Time-Delay Systems: Analysis and Control}.
\newblock Cham: Birkh{\"a}user, 2014.

\bibitem{Artstein1982}
Z.~Artstein, ``Linear systems with delayed controls: A reduction,'' {\em IEEE
  Transactions on Automatic Control}, vol.~27, no.~4, pp.~869--879, 1982.

\bibitem{Krstic2009}
M.~Krstic, {\em Delay Compensation for Nonlinear, Adaptive, and {PDE} Systems}.
\newblock Boston: Birkh{\"a}user, 2009.

\bibitem{Gu2003}
K.~Gu, V.~L. Kharitonov, and J.~Chen, {\em Stability of Time-Delay Systems}.
\newblock Boston: Birkh{\"a}user, 2003.

\bibitem{Kyrychko2010}
Y.~N. Kyrychko and S.~J. Hogan, ``On the use of delay equations in engineering
  applications,'' {\em Journal of Vibration and Control}, vol.~16, no.~7--8,
  pp.~943--960, 2010.

\bibitem{prajna2004safety}
S.~Prajna and A.~Jadbabaie, ``Safety verification of hybrid systems using
  barrier certificates,'' in {\em Proceedings of International Workshop on
  Hybrid Systems: Computation and Control}, pp.~477--492, Springer, 2004.

\bibitem{santoyo2021barrier}
C.~Santoyo, M.~Dutreix, and S.~Coogan, ``A barrier function approach to
  finite-time stochastic system verification and control,'' {\em Automatica},
  vol.~125, 2021.

\bibitem{nejati2024context}
A.~Nejati, S.~Prakash~Nayak, and A.-K. Schmuck, ``Context-triggered games for
  reactive synthesis over stochastic systems via control barrier
  certificates,'' in {\em Proceedings of the 27th ACM International Conference
  on Hybrid Systems: Computation and Control}, pp.~1--12, 2024.

\bibitem{borrmann2015control}
U.~Borrmann, L.~Wang, A.~D. Ames, and M.~Egerstedt, ``Control barrier
  certificates for safe swarm behavior,'' {\em IFAC-PapersOnLine}, vol.~48,
  no.~27, pp.~68--73, 2015.

\bibitem{zaker2024compositional}
M.~Zaker, O.~Akbarzadeh, B.~Samari, and A.~Lavaei, ``Compositional design of
  safety controllers for large-scale stochastic hybrid systems,'' {\em
  Automatica}, vol.~190, 2026.

\bibitem{ames2019control}
A.~D. Ames, S.~Coogan, M.~Egerstedt, G.~Notomista, K.~Sreenath, and P.~Tabuada,
  ``Control {B}arrier {F}unctions: {T}heory and {A}pplications,'' in {\em
  European control conference (ECC)}, pp.~3420--3431, 2019.

\bibitem{lavaei2022automated}
A.~Lavaei, S.~Soudjani, A.~Abate, and M.~Zamani, ``Automated verification and
  synthesis of stochastic hybrid systems: {A} survey,'' {\em Automatica},
  vol.~146, 2022.

\bibitem{Jadbabaie_delay}
S.~Prajna and A.~Jadbabaie, ``Methods for safety verification of time-delay
  systems,'' in {\em IEEE Conference on Decision and Control (CDC)},
  pp.~4348--4353, 2005.

\bibitem{Ames_Safety}
G.~Orosz and A.~D. Ames, ``Safety functionals for time delay systems,'' in {\em
  American Control Conference (ACC)}, pp.~4374--4379, 2019.

\bibitem{Ames_Safety_1}
A.~K. Kiss, T.~G. Molnar, A.~D. Ames, and G.~Orosz, ``Control barrier
  functionals: Safety-critical control for time delay systems,'' {\em
  International Journal of Robust and Nonlinear Control}, vol.~33, no.~12,
  pp.~7282--7309, 2023.

\bibitem{liu2023safety}
W.~Liu, Y.~Bai, L.~Jiao, and N.~Zhan, ``Safety guarantee for time-delay systems
  with disturbances,'' {\em Science China Information Sciences}, vol.~66,
  no.~3, 2023.

\bibitem{Papachristodoulou_Convex}
A.~Papachristodoulou, ``Robust stabilization of nonlinear time delay systems
  using convex optimization,'' in {\em IEEE Conference on Decision and Control
  (CDC)}, pp.~5788--5793, 2005.

\bibitem{REN2022110563}
W.~Ren, R.~M. Jungers, and D.~V. Dimarogonas, ``Razumikhin and {K}rasovskii
  approaches for safe stabilization,'' {\em Automatica}, vol.~146, 2022.

\bibitem{Fridman}
E.~Fridman and U.~Shaked, ``Stability and guaranteed cost control of uncertain
  discrete delay systems,'' {\em International Journal of Control}, vol.~78,
  no.~4, pp.~235--246, 2005.

\bibitem{Hou2013model}
Z.~S. Hou and Z.~Wang, ``From {M}odel-based {C}ontrol to {D}ata-driven
  {C}ontrol: Survey, {C}lassification and {P}erspective,'' {\em Information
  Sciences}, vol.~235, pp.~3--35, 2013.

\bibitem{dorfler2022bridging}
F.~D{\"o}rfler, J.~Coulson, and I.~Markovsky, ``Bridging direct and indirect
  data-driven control formulations via regularizations and relaxations,'' {\em
  IEEE Transactions on Automatic Control}, vol.~68, no.~2, pp.~883--897, 2023.

\bibitem{nejati2023data}
A.~Nejati and M.~Zamani, ``Data-driven {S}ynthesis of {S}afety {C}ontrollers
  via {M}ultiple {C}ontrol {B}arrier {C}ertificates,'' {\em IEEE Control
  Systems Letters}, vol.~7, pp.~2497--2502, 2023.

\bibitem{bisoffi2020controller}
A.~Bisoffi, C.~De~Persis, and P.~Tesi, ``Controller {D}esign for {R}obust
  {I}nvariance from {N}oisy {D}ata,'' {\em IEEE Transactions on Automatic
  Control}, vol.~68, no.~1, pp.~636--643, 2023.

\bibitem{nejati2022data}
A.~Nejati, B.~Zhong, M.~Caccamo, and M.~Zamani, ``Data-driven controller
  synthesis of unknown nonlinear polynomial systems via control barrier
  certificates,'' in {\em Proceedings of the 4th Annual Learning for Dynamics
  and Control Conference}, vol.~168 of {\em Proceedings of Machine Learning
  Research}, pp.~763--776, 2022.

\bibitem{Akbarzadeh_LCSS}
B.~Samari, O.~Akbarzadeh, M.~Zaker, and A.~Lavaei, ``From a {S}ingle
  {T}rajectory to {S}afety {C}ontroller {S}ynthesis of {D}iscrete-{T}ime
  {N}onlinear {P}olynomial {S}ystems,'' {\em IEEE Control Systems Letters},
  vol.~8, pp.~3123--3128, 2024.

\bibitem{Fridman-data}
J.~G. Rueda-Escobedo, E.~Fridman, and J.~Schiffer, ``Data-driven control for
  linear discrete-time delay systems,'' {\em IEEE Transactions on Automatic
  Control}, vol.~67, no.~7, pp.~3321--3336, 2022.

\bibitem{Kong}
P.~Kong, M.~Wang, H.~Yan, Z.~Li, and Y.~Lv, ``Data-driven control for linear
  discrete-time systems with time-varying delays,'' {\em International Journal
  of Systems Science}, vol.~57, no.~3, pp.~635--646, 2026.

\bibitem{calafiore2006scenario}
G.~C. Calafiore and M.~C. Campi, ``The {S}cenario {A}pproach to {R}obust
  {C}ontrol {D}esign,'' {\em IEEE Transactions on Automatic Control}, vol.~51,
  no.~5, pp.~742--753, 2006.

\bibitem{esfahani2014performance}
P.~Mohajerin~Esfahani, T.~Sutter, and J.~Lygeros, ``Performance {B}ounds for
  the {S}cenario {A}pproach and an {E}xtension to a {C}lass of {N}on-{C}onvex
  {P}rograms,'' {\em IEEE Transactions on Automatic Control}, vol.~60, no.~1,
  pp.~46--58, 2015.

\bibitem{margellos2014road}
K.~Margellos, P.~Goulart, and J.~Lygeros, ``On the road between robust
  optimization and the scenario approach for chance constrained optimization
  problems,'' {\em IEEE Transactions on Automatic Control}, vol.~59, no.~8,
  pp.~2258--2263, 2014.

\bibitem{akbarzadeh2024learning}
O.~Akbarzadeh, M.~H. Ashoori, and A.~Lavaei, ``Learning {R}obust {S}afety
  {C}ontrollers for {U}ncertain {I}nput-{A}ffine {P}olynomial {S}ystems,'' in
  {\em IEEE 64th Conference on Decision and Control (CDC)}, pp.~5788--5793,
  2025.

\bibitem{guo2022datadriven}
M.~Guo, C.~De~Persis, and P.~Tesi, ``Data-driven stabilization of nonlinear
  polynomial systems with noisy data,'' {\em IEEE Transactions on Automatic
  Control}, vol.~67, no.~8, pp.~4210--4217, 2022.

\bibitem{Donkers2012}
M.~Donkers, W.~Heemels, D.~Bernardini, A.~Bemporad, and V.~Shneer, ``Stability
  analysis of stochastic networked control systems,'' {\em Automatica},
  vol.~48, no.~5, pp.~917--925, 2012.

\bibitem{Papachristodoulou_time_delay}
A.~Papachristodoulou, ``Analysis of nonlinear time-delay systems using the sum
  of squares decomposition,'' in {\em Proceedings of the American Control
  Conference}, pp.~4153--4158, 2004.

\bibitem{bhatia1995cauchy}
R.~Bhatia and C.~Davis, ``A {C}auchy-{S}chwarz inequality for operators with
  applications,'' {\em Linear Algebra and its Applications}, vol.~223--224,
  pp.~119--129, 1995.

\bibitem{young1912classes}
W.~H. Young, ``On classes of summable functions and their {F}ourier series,''
  {\em Proceedings of the Royal Society of London. Series A}, vol.~87, no.~594,
  pp.~225--229, 1912.

\bibitem{zhang2006schur}
F.~Zhang, {\em The Schur complement and its applications}, vol.~4.
\newblock Springer Science \& Business Media, 2006.

\bibitem{caverly2019lmi}
R.~J. Caverly and J.~R. Forbes, ``{LMI} properties and applications in systems,
  stability, and control theory,'' {\em arXiv: 1903.08599}, 2019.

\bibitem{WILLEMS2005325}
J.~C. Willems, P.~Rapisarda, I.~Markovsky, and B.~L. {De Moor}, ``A {N}ote on
  {P}ersistency of {E}xcitation,'' {\em Systems \& Control Letters}, vol.~54,
  no.~4, pp.~325--329, 2005.

\bibitem{Alsalti_Persistently}
M.~Alsalti, V.~G. Lopez, and M.~A. M{\"u}ller, ``On the design of persistently
  exciting inputs for data-driven control of linear and nonlinear systems,''
  {\em IEEE Control Systems Letters}, vol.~7, pp.~2629--2634, 2023.

\bibitem{prajna2004sostools}
S.~Prajna, A.~Papachristodoulou, P.~Seiler, and P.~A. Parrilo, ``{SOSTOOLS}:
  Control {A}pplications and {N}ew {D}evelopments,'' in {\em IEEE International
  Conference on Robotics and Automation}, pp.~315--320, 2004.

\bibitem{mosek}
{MOSEK ApS}, {\em The {MOSEK} {O}ptimization {T}oolbox for {MATLAB} {M}anual.
  Version 10.1.}, 2025.

\bibitem{Tabuada-Jet}
A.~Anta and P.~Tabuada, ``To sample or not to sample: Self-triggered control
  for nonlinear systems,'' {\em IEEE Transactions on Automatic Control},
  vol.~55, no.~9, pp.~2030--2042, 2010.

\bibitem{khalil2002control}
H.~K. Khalil, {\em Nonlinear {S}ystems}.
\newblock Prentice Hall, 3~ed., 2002.

\end{thebibliography}
	\vspace{-0.5cm}
	
	\begin{IEEEbiography}[{\includegraphics[width=1in,height=1.25in,clip,keepaspectratio]{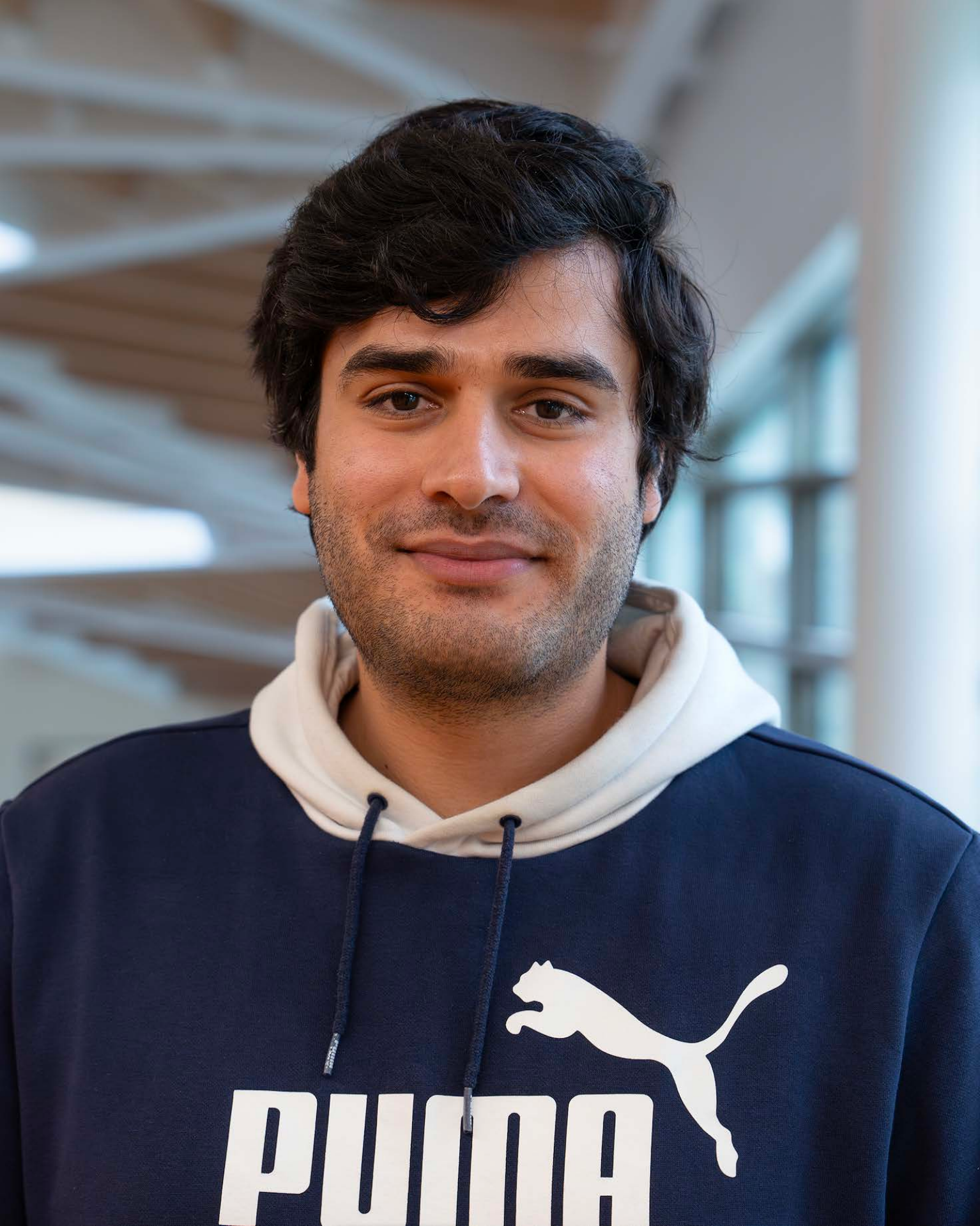}}]{Omid Akbarzadeh}~(Student Member, IEEE) is currently a PhD student in the School of Computing at Newcastle University, U.K. His academic journey commenced at Shiraz University, where he obtained a Bachelor of Science in Electrical and Electronic Engineering. Following this, he pursued a master's degree in Communications and Computer Network Engineering (CCNE) at the Polytechnic University of Turin, Italy (Politecnico di Torino). His research interests include safe cyber-physical systems, communication networks, data-driven approaches, and formal control.
	\end{IEEEbiography}	\vspace{-0.8cm}
	
	\begin{IEEEbiography}[{\includegraphics[width=1in,height=1.3in,clip,keepaspectratio]{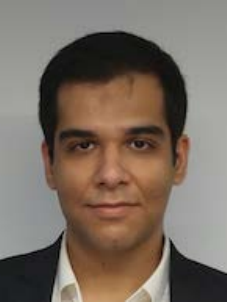}}]{MohammadHossein Ashoori}~(Student Member, IEEE) received his B.Sc. and M.Sc. degrees in Electrical Engineering from Sharif University of Technology (SUT), Tehran, Iran, in 2019 and 2022, respectively. He is currently pursuing his PhD in
		the School of Computing at Newcastle University,
		UK.  His research
		interests include cyber-physical systems (CPS), computer vision, and digital signal processing.
	\end{IEEEbiography}	\vspace{-0.8cm}
	
\begin{IEEEbiography}[{\includegraphics[width=1in,height=1.25in,clip,keepaspectratio]{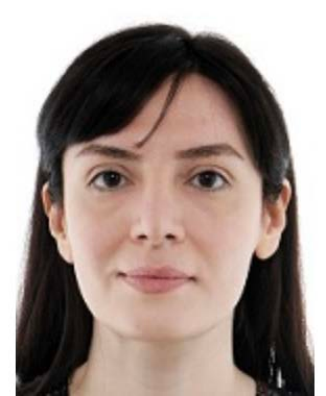}}]{Amy Nejati}~(M'18--SM'25) is an Assistant Professor in the School of Computing at Newcastle University in the United Kingdom. Prior to this, she was a Postdoctoral Associate at the Max Planck Institute for Software Systems in Germany from July 2023 to May 2024. She also served as a Senior Researcher in the Computer Science Department at the Ludwig Maximilian University of Munich (LMU) from November 2022 to June 2023. She received her PhD in Electrical Engineering from the Technical University of Munich (TUM) in 2023. She received her B.Sc. and M.Sc. degrees both in Electrical Engineering. Her line of research mainly focuses on developing efficient (data-driven) techniques to design and control highly reliable autonomous systems while providing mathematical guarantees. She was selected as a Best Repeatability Prize Finalist at ACM HSCC 2025 and as one of the CPS Rising Stars 2024.\end{IEEEbiography}\vspace{-0.8cm}

\begin{IEEEbiography}
	[{\includegraphics[width=1in,height=1.25in,clip,keepaspectratio]{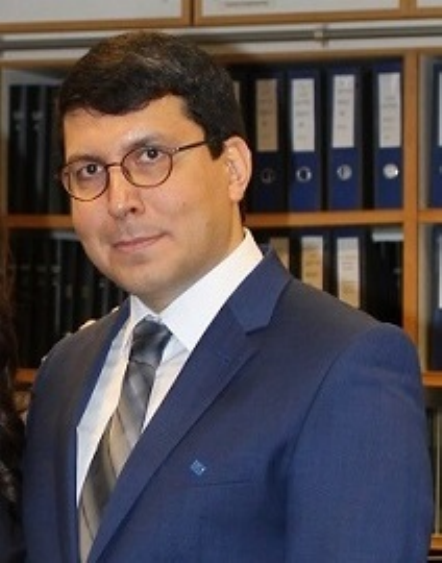}}]{Abolfazl Lavaei}~(M'17--SM'22)
	is an Assistant Professor in the School of Computing at Newcastle University, United Kingdom. Between January 2021 and July 2022, he was a Postdoctoral Associate in the Institute for Dynamic Systems and Control at ETH Zurich, Switzerland. He was also a Postdoctoral Researcher in the Department of Computer Science at LMU Munich, Germany, between November 2019 and January 2021. He received the Ph.D. degree in Electrical Engineering from the Technical University of Munich (TUM), Germany, in 2019. He obtained the M.Sc. degree in Aerospace Engineering with a specialization in Flight Dynamics and Control from the University of Tehran (UT), Iran, in 2014. He is the recipient of several international awards in recognition of his work including  Best Repeatability Prize (Finalist) at the ACM HSCC 2025, IFAC ADHS 2024, and IFAC ADHS 2021, HSCC Best Demo/Poster Awards 2022 and 2020, IFAC Young Author Award Finalist 2019, and Best Graduate Student Award 2014 at the University of Tehran with the full GPA (20/20). His research interests revolve around the intersection of Control Theory, Formal Methods, and Statistical Learning Theory.
\end{IEEEbiography}
	
\end{document}